\documentclass[11pt]{amsart}
\usepackage[left=3.cm, right=3.cm, top=3.5cm, bottom=3.5cm]{geometry}
\usepackage{color}
\usepackage{pgf}
\usepackage{mathtools}
\usepackage{enumerate}
\usepackage{tikz}

\usetikzlibrary{patterns}
\usetikzlibrary{calc}
\usetikzlibrary{arrows,shapes,positioning}
\usetikzlibrary{decorations.markings}
\usetikzlibrary[matrix] 

\tikzstyle directed=[postaction={decorate,decoration={markings,
    mark=at position .65 with {\arrow[arrowstyle]{latex}}}}]

\tikzstyle arrowstyle=[scale=1]

\usepackage{amssymb}
\usepackage{cite}

\newtheorem{thm}{Theorem}[section]
\newtheorem{cor}[thm]{Corollary}
\newtheorem{conj}[thm]{Conjecture}
\newtheorem{prop}[thm]{Proposition}
\newtheorem{lem}[thm]{Lemma}
\theoremstyle{definition}

\newtheorem{example}[thm]{Example}
\newtheorem{definition}[thm]{Definition}
\newtheorem{remark}[thm]{Remark}
\numberwithin{equation}{section}
\numberwithin{figure}{section}

\DeclareMathOperator{\sgn}{\operatorname{sgn}}

\DeclareMathOperator{\Res}{\operatorname{Res}}
\newcommand{\Wr}{\operatorname{Wr}}
\newcommand{\N}{\mathbb{N}}
\newcommand{\Nz}{{\N_0}}
\newcommand{\Z}{\mathbb{Z}}
\newcommand{\C}{\mathbb{C}}
\newcommand{\R}{\mathbb{R}}

\newcommand{\ZkZ}{\Z/k\Z}
\newcommand{\bzero}{\boldsymbol{0}}
\newcommand{\balpha}{{\boldsymbol{\alpha}}}
\newcommand{\bbeta}{{\boldsymbol{\beta}}}
\newcommand{\bbe}{\bbeta}
\newcommand{\bgamma}{{\boldsymbol{\gamma}}}

\newcommand{\ba}{{\boldsymbol{a}}}
\newcommand{\bw}{{\boldsymbol{w}}}
\newcommand{\bdf}{{\boldsymbol{f}}}

\newcommand{\hw}{\hat{w}}
\newcommand{\ha}{\hat{a}}

\newcommand{\bnu}{{\boldsymbol{\nu}}}
\newcommand{\be}{{\boldsymbol{e}}}
\newcommand{\hbnu}{{\hat{\bnu}}}
\newcommand{\hnu}{{\hat{\nu}}}
\newcommand{\bM}{{\boldsymbol{M}}}
\newcommand{\hbM}{\hat{\bM}}
\newcommand{\bs}{{\boldsymbol{s}}}
\newcommand{\bmu}{{\boldsymbol{\mu}}}

\newcommand{\bpi}{{\boldsymbol{\pi}}}
\newcommand{\bsigma}{{\boldsymbol{\sigma}}}

\newcommand{\tw}{{\tilde{w}}}

\newcommand{\bp}{{\boldsymbol{p}}}

\newcommand{\supth}{{}^{\rm{th}}}

\newcommand{\lp}{\left(}
\newcommand{\rp}{\right)}
\renewcommand{\H}{H}
\newcommand{\cL}{L}

\newcommand{\cM}{\mathcal{M}}
\newcommand{\cZ}{\mathcal{Z}}
\newcommand{\hcZ}{\hat{\cZ}}
\newcommand{\cZp}{\cZ^{p}}
\newcommand{\hcZp}{\hcZ^{p}}
\newcommand{\cZo}{\cZ^{\text{odd}}}
\newcommand{\hcZo}{\hcZ^{\text{odd}}}

\newcommand{\Zp}{\Z^{p}}

\newcommand{\cS}{\mathcal{S}}

\newcommand{\Xik}{\Xi_{k}}
\newcommand{\hM}{\hat{M}}

\newcommand{\ddz}{\frac{{\rm d}}{{\rm d}z}}

\newcommand{\EAWeyln}{\tilde{A}^{(1)}_{2n}}

\renewcommand{\th}{\widetilde{H}}

\newcommand{\h}{H}
\newcommand{\Pfour}{{\rm P_{IV}}}
\newcommand{\Pfive}{{\rm P_{V}}}
\newcommand{\Psix}{{\rm P_{VI}}}
\newcommand{\PI}{$\mathrm{P}_{\mathrm{I}}$}
\renewcommand{\PII}{$\mathrm{P}_{\mathrm{II}}$}
\newcommand{\PIII}{$\mathrm{P}_{\mathrm{III}}$}
\newcommand{\PIV}{$\mathrm{P}_{\mathrm{IV}}$}

\renewcommand{\boxdot}{{\ \clap{\raise0.25ex\hbox{$\bullet$}}\clap{$\square$}\ }}
\newcommand{\emptybox}{\hbox{$\square$}}

\newcommand{\p}{Painlev\'{e}}

\newcommand{\bk}{B\"acklund}

\renewcommand{\d}{{\rm d}}

\renewcommand{\a}{\alpha}
\renewcommand{\b}{\beta}
\newcommand{\ex}[1]{\exp(#1)}

\newcommand{\cc}[1]{c_{#1}}

\newcommand{\odes}{ordinary differential equations}

\newcommand{\rd}[1]{{\color{red}#1}}
\newcommand{\bl}[1]{{\color{blue}#1}}

\begin{document}

\title[Rational solutions of $A_{2n}$-Painlev\'{e} systems]{Complete classification of  rational solutions of $A_{2n}$-Painlev\'{e} systems.}

\author{David G\'omez-Ullate}
\address{Escuela Superior de Ingenier\'ia, Universidad de C\'adiz, 11519 Puerto Real, Spain.}
\address{Departamento de F\'isica Te\'orica, Universidad Complutense de
  Madrid, 28040 Madrid, Spain.}

\author{Yves Grandati}
\address{ LCP A2MC, Universit\'{e} de Lorraine, 1 Bd Arago, 57078 Metz, Cedex 3,
  France.}
\author{Robert Milson}
\address{Department of Mathematics and Statistics, Dalhousie University,
  Halifax, NS, B3H 3J5, Canada.}
\email{david.gomezullate@uca.es, grandati@univ-metz.fr,   rmilson@dal.ca}
\begin{abstract}
We provide a complete classification and an explicit representation of rational solutions to the fourth Painlev\'e equation $\Pfour$ and its higher order generalizations known as the $A_{2n}$-\p\ or  Noumi-Yamada systems. The construction of solutions makes use of the theory of cyclic dressing chains of Schr\"odinger operators. Studying the local expansions of the solutions around their singularities we find that some coefficients in their Laurent expansion must vanish, which express precisely the conditions of  trivial monodromy of the associated potentials. The characterization of trivial monodromy potentials with quadratic growth implies that all rational solutions can be expressed as Wronskian determinants of suitably chosen sequences of Hermite polynomials. The main classification result states that every rational solution to the $A_{2n}$-\p\ system corresponds to a cycle of Maya diagrams, which can be indexed by an oddly coloured integer sequence.  Finally, we establish the link with the standard approach to building rational solutions, based on applying B\"acklund transformations  on seed solutions, by providing a representation for the symmetry group action on coloured sequences and Maya cycles.

\bigskip
\noindent \textbf{Keywords.} Painlev\'e equations, Noumi-Yamada
systems, rational solutions, Darboux dressing chains, Maya diagrams,
Wronskian determinants, Hermite polynomials.

\end{abstract}

\maketitle
\setcounter{tocdepth}{2}
\tableofcontents
\section{Introduction}\label{sec:intro}

The solutions of \p\ equations are  considered to be the nonlinear analogues of special functions,\cite{clarkson2003painleve,refFIKNbook,gromak2008painleve}. 
In general, they are transcendental functions, but for special values of the parameters, \p\ equations (except the first one) possess solutions that can be expressed via rational or special functions. For a review of rational solutions to \p\ equations, see the recent book by Van Assche, \cite{refWVAbook}.

In this paper we focus on the rational solutions of Painlev\'e's fourth equation (\PIV) 
 \eqref{eq:P4} 
\begin{equation}\label{eq:P4}
{\rm P_{IV}}:\qquad u''=\frac{(u')^2}{2u} 
+\frac32u^3+4zu^2+2(z^2-\a)u+\frac{\beta}{u},\qquad \alpha,\beta\in\mathbb C,
\end{equation}
and its higher order generalizations, known as  the $A_{2n}$-\p\ or  Noumi-Yamada systems. 

Lukasevich \cite{lukashevich1967theory} found by direct inspection the first few rational solutions of $\Pfour$. Okamoto \cite{okamoto1987studies3} developed the theory of symmetry transformations of this equation, finding a Hamiltonian structure, birational canonical transformations, parameters for which rational solutions exist and some special solutions that now bear his name.

The following system appeared in Bureau's survey of systems with
  fixed critical points \cite{bureau92} as an example that could be
  reduced to the scalar 2nd order equation $\Pfour$:
\begin{align}\label{eq:P4system}
&f_0' + f_0(f_1-f_2) = \a_0, \nonumber\\
  {\rm sP_{IV}}:\qquad
  &         f_1' + f_1(f_2-f_0) = \a_1,\\
&f_2' + f_2(f_0-f_1) = \a_2, \nonumber 
\end{align}
with $'\equiv \d/\d z$ and $\a_j$, $j=0,1,2$ constants,
subject to the normalization conditions
\begin{equation}\label{eq:P4normalization}
f_0+f_1+f_2=z,\qquad \a_0+\a_1+\a_2=1.
\end{equation}
This system, nowadays commonly referred to as the symmetric form of
\PIV, also admits a symmetry group of \bk\ transformations
\cite{noumi1999symmetries} acting on the tuple of solutions and
parameters $(f_0,f_1,f_2 |\a_0,\a_1,\a_2)$. The symmetry group is
generated by the operators
$\{ \boldsymbol{\pi},\textbf{s}_0, \textbf{s}_1, \textbf{s}_2\}$ with
the following action on the tuple $(f_0,f_1,f_2 |\a_0,\a_1,\a_2)$:
\begin{equation}
  \label{eq:BT}
\begin{aligned}
  \bs_0(f_0) &= f_0,\quad \bs_0(f_1) = f_1 -
  \frac{\alpha_0}{f_0},\quad
  \bs_0(f_2) = f_2 +\frac{\alpha_0}{f_0}\\
  \bs_0(\alpha_0) &= -\alpha_0,\quad \bs_0(\alpha_1) =
  \alpha_1+\alpha_0,\quad \bs_0(\alpha_2) = \alpha_2+\alpha_0\\
  \bpi(f_0) &= f_1,\quad \bpi(f_1) = f_2,\quad \bpi(f_2) = f_0\\
  \bpi(\a_0) &= \a_1,\quad \bpi(\a_1) = \a_2,\quad \bpi(\a_2) = \a_0
\end{aligned}
\end{equation}
with the action of $\bs_1, \bs_2$ obtained by cyclically permuting the
indices in the action of $\bs_0$.  Formally, the symmetry group is
isomorphic to the extended\footnote{The corresponding affine Weyl
  group is the subgroup that omits the $\pi$ transformation
  \cite{saito}.} affine Weyl group $A_2^{(1)}$, because the above
generators obey the defining relations
\begin{equation}\label{eq:group}
  \bs_i ^2\equiv 1,\quad (\bs_i \bs_{i+1})^{3}\equiv 1,\quad \bpi \bs_i
  \equiv\bs_{i+1}\bpi,\quad \bpi^{3}\equiv 1,\quad i\equiv 0,1,2\mod 3.
\end{equation}

Noumi and Yamada soon realized that the structure of \eqref{eq:P4system} can be generalized to any number of equations \cite{noumi1998higher}, leading to the $A_{N}$-\p\ or the Noumi-Yamada system. Systems with an even or odd number of equations have a rather different behaviour, and we restrict in this paper to the analysis of the $A_{2n}$-\p\ system, whose equations are given by
\begin{equation}\label{eq:Ansystem}
  \mbox{$A_{2n}$-Painlev\'e:}\qquad   f_i'+f_i \left( \sum_{j=1}^n
    f_{i+2j-1} - \sum_{j=1}^n f_{i+2j}  
  \right)=\a_i,\qquad i=0,\dots,2n \mod (2n+1)  
\end{equation}
subject to the normalization conditions
\begin{equation}
  \label{eq:alpha112}
f_0+\dots+f_{2n}= z,\qquad \a_0+\cdots + \a_{2n}=1.
\end{equation}

An equivalent system of Painlev\'e type equations was constructed earlier by Veselov
and Shabat \cite{veselov1993dressing} for which Adler \cite{adler93} explicitly described the affine Weyl group symmetry and gave its geometric interpretation. This system can be considered the natural higher order generalization
of $\,\rm{s}\Pfour$ (which corresponds to $n=1$), since its symmetry
group is the extended affine Weyl group $\EAWeyln$, acting by \bk\
transformations as in \eqref{eq:BT}. The system passes the
Painlev\'e-Kowalevskaya test, \cite{veselov1993dressing}.

The standard technique to construct rational solutions of \eqref{eq:Ansystem} is to start from a number of very simple rational \textit{seed solutions}, and successively apply the \bk\ transformations \eqref{eq:BT} to generate new solutions, which are rational by construction. However, this method does not produce \textit{per se} explicit representations of the solutions. For this reason, other more explicit representations have been investigated, most notably via recursion relations \cite{okamoto1987studies3,fukutani2000special}, determinantal representations \cite{kajiwara1998determinant,noumi1999symmetries} or Schur functions, exploiting suitable reductions of the KP hierarchy in  Sato's theory of integrable systems, \cite{tsuda2012kp,tsuda2004universal}. Perhaps the simplest representation of the rational solutions of \PIV\ is via Wronskian determinants of certain sequences of Hermite polynomials:
\begin{eqnarray}
H_{m,n}(z)&=&\Wr (H_m,H_{m+1},\dots,H_{m+n-1}), \label{eq:GenH}\\
Q_{m,n}(z)&=&\Wr (H_1,H_4,\dots,H_{1+3(m-1)},H_2,H_5,\dots,H_{2+3(n-1)}), \label{eq:Oka}
\end{eqnarray}
which are known as \textit{generalized Hermite} and \textit{generalized Okamoto} polynomials, respectively.

Regarding higher order systems, rational solutions of $A_{4}$-\p\ have been investigated in \cite{filipuk2008symmetric,matsuda2012rational} and classified recently in \cite{clarkson2020cyclic}, which lays the ground for the construction of solutions in this paper. For systems of arbitrary order $N$, Tsuda \cite{tsuda2005universal} has described one special family of solutions in terms of Schur functions associated to $N$-reduced partitions, which can be regarded as a generalization of \eqref{eq:Oka}.
Indeed, the families \eqref{eq:GenH} and \eqref{eq:Oka} can be generalized to the higher order system \eqref{eq:Ansystem}, but they represent only a small part of all the solutions, those corresponding to the minimal and maximal shifts.

The special polynomials associated with rational solutions of \p\ equations have attracted much interest for various reasons.
First, they appear in a number of applications, in connection with  random matrix theory \cite{refFW01,chen2006painleve}, supersymmetric quantum mechanics \cite{bermudez2012complexb,marquette2016,Novokshenov18,bermudez2011supersymmetric}, vortex dynamics with quadrupole background flow \cite{Clarkson2009vortices},  recurrence relations for orthogonal polynomials \cite{clarkson2014relationship,refWVAbook}, exceptional orthogonal polynomials \cite{garcia2016bochner} or rational-oscillatory solutions of the defocusing nonlinear Schr{\"o}dinger equation \cite{clarkson2006special}.

Second, the complex zeros of these special polynomials form remarkably regular patterns in the complex plane, as it has been mostly studied by Clarkson, \cite{clarkson2003fourth}. The zeros of generalized Hermite polynomials \eqref{eq:GenH} form rectangular patterns, and for large $m,n$ with $m/n$ fixed  they fill densely a curvilinear rectangle whose boundary is described by Buckingham \cite{BuckPIV} using the steepest descent method for a Riemann-Hilbert problem. The distribution of these zeros is also studied recently by Masoero and Raffolsen in other asymptotic regimes, \cite{refMR18,masoero2019roots}. 
The roots of generalized Okamoto polynomials form patterns that combine rectangular and triangular filled regions, and recently Buckingham and Miller \cite{buckingham2020large} have extended their analysis to provide a rigorous description of the boundaries. Remarkably, the zeros and poles of rational solutions to higher order $A_{2n}$-\p\ systems show much richer structures, which so far have only been investigated numerically, \cite{clarkson2020cyclic}.

Construction methods are able to prove existence of rational solutions and equivalence of different representations, but the question of establishing that \textit{all of the rational solutions} are obtained is much harder, and has only been addressed in very few papers.  Parameters for which rational solutions exist have been identified by Murata \cite{murata1985i} for \PIV\ and by Kitaev, Law and McLeod \cite{kitaev1994rational}  for $\Pfive$. These results are obtained by direct computation on local expansions, and they do not scale well to higher order systems due to increasing complexity and branching. By contrast, Veselov was able to establish that rational solutions of $A_{2n}$-\p\ are in one-to-one correspondence with Schr\"odinger operators whose potentials have quadratic growth at infinity and trivial monodromy. His paper \cite{veselov2001stieltjes}, which has received comparatively less attention, is the basis for the characterization of rational solutions performed in this work, together with \cite{refDG86,oblomkov1999monodromy}.

The factorization method in quantum mechanics, first introduced by
  Schr\"{o}dinger \cite{schrodinger} and presented in its general form
  by Infeld and Hull \cite{infeld}, rests implicitily on an operator
  factorization transformation introduced by Darboux in his study of
  surfaces\cite{darboux}.  In \cite{veselov1993dressing} Veselov and
  Shabat introduced the notion of a \textit{dressing chain}, which is a sequence
  of a Darboux transformations that after a finite number of steps
  returns to the initial potential shifted by an additive constant.
  The connection to \PIV\ was already noted in this paper, and indeed
  the dressing chain is formally equivalent to the Noumi-Yamada
  system.  Adler considered various generalization of the dressing
  chain that allow the construction of other Painlev\'e equations
  \cite{adler1994nonlinear}.  Adler, in a somewhat
  earlier paper\cite{adler93}, also discovered the connection to
  symmetry groups and B\"acklund transformations.

As mentioned in \cite{WilloxHietarinta}, our aim in this paper is to combine the strength of the $\tau$-function and geometric approach of the japanese school \cite{noumi1999symmetries, okamoto1987studies3, tsuda2005universal,tsuda2012kp, fukutani2000special,noumi2004painleve} with that of the dressing chains and trivial monodromy approach of the russian school \cite{adler1994nonlinear,veselov1993dressing, veselov2001stieltjes, oblomkov1999monodromy}, to attain our goal of giving a complete classification of the rational solutions to higher order \p\ systems. Given the breadth of both points of view, beyond this goal there is much to be learnt from their common interplay.

The paper is organized as follows.  In Section~\ref{sec:preliminaries}
we recall the equivalence between the $A_{2n}$-\p\ system and cyclic
dressing chains of Darboux transformations obtained as factorizations
of Schr\"{o}dinger operators, \cite{veselov1993dressing,
  adler1994nonlinear}.  Section~\ref{sec:Maya} introduces the class of
rational extensions of the harmonic oscillator and identifies them as
the only potentials with quadratic growth at infinity and trivial
monodromy \cite{oblomkov1999monodromy}. It also introduces Hermite
pseudo-Wronskians indexed by Maya diagrams, and recalls some of their
basic properties \cite{gomez2016durfee}. The main result in this
section is Proposition~\ref{prop:seedfunc} that provides all
quasi-rational eigenfunctions of Schr\"odinger operators belonging to
that class of potentials. In Section~\ref{sec:veselov} we follow the
work of Veselov\cite{veselov2001stieltjes} on rational solutions to
odd-cyclic dressing chains. Studying the Laurent expansions of these
solutions, the constraints imposed on the coefficients of the
expansion are identified precisely as the conditions that express
trivial monodromy of the associated potentials of the chain. The main
result in this section is Theorem~\ref{thm:characterization}, which
establishes that all rational solutions of an odd-cyclic dressing
chain must necessarily be expressible as Wronskian determinants of
Hermite polynomials. Section~\ref{sec:cycles} studies cycles of Maya
diagrams and introduces all the necessary concepts (genus,
interlacing, block coordinates) to achieve a complete classification,
which is described in
Proposition~\ref{thm:Mp}. Section~\ref{sec:classification} uses all
the previously derived results to state and prove the main
Theorem~\ref{thm:main} on the classification of rational solutions of
odd-cyclic dressing chains. The result not only provides a full
classification, but also allows for an explicit representation of all
solutions in terms of oddly coloured sequences. To illustrate this,
some explicit examples are given in \S \ref{sec:examples}. Finally, in
Section~\ref{sec:sym} we study the action of the symmetry group
\eqref{eq:BT} on the representation of the solutions given by
Theorem~\ref{thm:main}, thus providing a connection with the standard
approach \cite{noumi2004painleve}. Possible extensions of this work
are discussed in Section~\ref{sec:summary}, where we formulate a
conjecture on the equivalent result for even cyclic chains.

\section{Higher order Painlev\'e equations and dressing chains}\label{sec:preliminaries}

The $A_{2n}$-\p\ system is the following set of $2n+1$ nonlinear differential equations for the functions $f_i=f_i(z)$ and parameters $\alpha_i\in\mathbb{C}$
\begin{equation}\label{eq:A2nsystem}
  f_i'+f_i \left( \sum_{j=1}^n f_{i+2j-1} - \sum_{j=1}^n f_{i+2j}
  \right)=\a_i,\qquad i=0,\dots,2n \mod (2n+1)
\end{equation}
subject to the normalization conditions
\begin{equation}
  \label{eq:alpha1}
f_0+\dots+f_{2n}= z,\qquad \a_0+\cdots + \a_{2n}=1.
\end{equation}

\begin{definition}
A \textit{rational solution} of the $A_{2n}$-\p\ system \eqref{eq:A2nsystem} is a tuple of functions and parameters $(f_0,\dots,f_{2n}|\a_0,\dots,\a_{2n})$ where $f_i=f_i(z)$ are rational functions of $z$.
\end{definition}

\begin{remark}
As we shall see later, for every solution of an $A_{2m}$-\p\ system, one can build an infinite number of degenerate solutions of a higher order  $A_{2n}$-\p\ system, with $n>m$. One of them corresponds to trivially setting some of the $f_i$ and $\alpha_i$ to zero, but many other non-trivial embeddings also exist.
\end{remark}

It will be convenient throughout the paper to work with a different set of functions and parameters, namely the set of functions that satisfy a Darboux dressing chain, which we define next.

\begin{definition}\label{def:wchain}
  A $(2n+1)$-cyclic \textit{dressing chain}  with
  shift $\Delta$ is a sequence of $2n+1$ functions $w_0,\ldots, w_{2n}$
  and complex numbers $a_0,\ldots,a_{2n}$ that satisfy
  the following coupled system of $2n+1$ Riccati-like \odes\
\begin{equation}
  \label{eq:wfchain}
  (w_i + w_{i+1})' + w_{i+1}^2 - w_i^2 = a_i ,\qquad
  i=0,1,\ldots, 2n \mod(2n+1)  
\end{equation}
subject to the condition
\begin{equation}
  \label{eq:Deltasumalpha}
 a_0 + \cdots + a_{2n}=-\Delta. 
\end{equation}

Note that by adding the $2n+1$ equations \eqref{eq:wfchain} we immediately obtain a
first integral of the system
\begin{equation}\label{eq:wsum}
\sum_{j=0}^{2n} w_j= \tfrac12z\sum_{j=0}^{2n} a_j=
-\tfrac12{\Delta}z.
\end{equation}

The system \eqref{eq:wfchain} has a a group of symmetries that will be discussed in Section~\ref{sec:sym}. For now, we will just observe that it is invariant under two obvious transformations of functions and parameters:
\begin{enumerate}
\item[i)] reversal symmetry
\begin{equation}
  \label{eq:reversal2}
  w_i \mapsto -w_{-i},\quad a_i \mapsto -a_{-i},\quad
  \Delta\mapsto -\Delta
\end{equation}
\item[ii)] cyclic symmetry
\begin{equation}
 w_i \mapsto w_{i + 1},\quad a_i \mapsto a_{i+1},\quad
\Delta \mapsto \Delta 
\end{equation}
for $i=0,\ldots 2n \mod(2n+1)$. 
\end{enumerate}
The equivalence between the $A_{2n}$-\p\ system \eqref{eq:A2nsystem} and the $(2n+1)$-cyclic dressing chain \eqref{eq:wfchain} is given by the following proposition.

\begin{prop}\label{prop:wtof}
  The tuple of functions and complex numbers
  $(w_0,\dots,w_{2n}|a_0,\dots,a_{2n})$ satisfy \eqref{eq:wfchain}
  \eqref{eq:wsum}, the relations of a $(2n+1)$-cyclic Darboux dressing
  chain with shift $\Delta$, if and only if the tuple
  $\left(f_0,\dots,f_{2n}\,\big|\,\a_0,\dots, \a_{2n}\right)$
  defined by 
  \begin{equation}\label{eq:wtof2}
    \begin{aligned}
      f_i(z)&= c \,(w_i + w_{i+1})\left(cz\right),\qquad
      i=0,\dots,2n\mod(2n+1),\\
      \a_i&=c^2 a_i,\\
      c^2&=-\frac{1}{\Delta}
    \end{aligned}
  \end{equation}
  satisfies the $A_{2n}$-\p\ system \eqref{eq:A2nsystem} subject to
  the normalization \eqref{eq:alpha1}.
\end{prop}
\begin{proof}
It suffices to invert the linear transformation
\begin{equation}\label{eq:wtof}
f_i=w_i+w_{i+1}, \qquad i=0,\dots,2n\mod(2n+1)
\end{equation}
to obtain
\begin{equation}\label{eq:ftow}
  w_i= \tfrac{1}{2} \sum_{j=0}^{2n} (-1)^j f_{i+j}, \qquad
  i=0,\dots,2n\mod(2n+1),
\end{equation}
which imply the relations
\begin{equation}\label{eq:wdif}
  w_{i+1}-w_i = \sum_{j=0}^{2n-1} (-1)^j f_{i+j+1}, \qquad
  i=0,\dots,2n\mod(2n+1). 
\end{equation}
Inserting \eqref{eq:wtof} and \eqref{eq:wdif} into the equations of the cyclic dressing chain \eqref{eq:wfchain} leads to the $A_{2n}$-\p\ system \eqref{eq:A2nsystem}. For any constant $c\in\mathbb C$, the scaling transformation
\[f_i\mapsto c f_i,\quad z\mapsto cz,\quad \a_i\mapsto c^2 \a_i \]
preserves the form of the equations \eqref{eq:A2nsystem}. The choice $c^2=-\frac{1}{\Delta}$ ensures that the normalization \eqref{eq:alpha1} always holds, for dressing chains with different shifts $\Delta$.
\end{proof}
\end{definition}

\subsection{Factorization chains of Schr\"odinger operators}\label{sec:schr}

We next recall the relation between dressing chains and sequences of
Schr\"odinger operators related by Darboux transformations.

Consider the following sequence of Schr\"{o}dinger operators
\begin{equation}\label{eq:Lseq}
  \cL_i = -D_z^2 + U_i ,\qquad D_z= \frac{{\rm d}}{{\rm d}z},\quad U_i=U_i(z),\quad
  i\in \Z
\end{equation}
where each operator is related to the next by a Darboux transformation,  i.e. by the following factorization
\begin{equation}
  \label{eq:Dxform}
  \begin{aligned}
    \cL_i &= (D_z + w_i)(-D_z + w_i)+\lambda_i, \quad w_i = w_i(z),\\
    \cL_{i+1} &= (-D_z + w_i)(D_z + w_i)+\lambda_i.
  \end{aligned}
\end{equation}
Eliminating the derivative terms we see that \eqref{eq:Dxform} is
equivalent to
\begin{equation}\label{eq:Riccati}
  w_i' + w_i^2 = U_i - \lambda_i,\quad -w_i' +w_i^2 = U_{i+1}- \lambda_i.
\end{equation}
Equivalently, we can characterize $w_i$ as the log-derivative of
$\psi_i$, the seed function of the Darboux transformation that maps
$\cL_i$ to $\cL_{i+1}$
\begin{equation}
  \label{eq:Lipsii}
  \cL_i\psi_i = \lambda_i\psi_i,\qquad\text{where } w_i = \frac{\psi_i'}{\psi_i}.
\end{equation}
Using \eqref{eq:Riccati}  the potentials of the
dressing chain are then related by
\begin{align}
  \label{Ui-Ui+1}
  U_{i}-U_{i+1} &= 2 w'_i, \\
  \label{Ui+Ui+1}
  U_i + U_{i+1} &= 2 w_i^2+2\lambda_i .
\end{align}
It follows that if \eqref{eq:Dxform} holds with non-constant $U_i$,
then the corresponding $w_i, \lambda_i$ are determined uniquely by the
potentials $U_i,\; i\in \Z$.

\begin{definition}
  We say that a sequence of Schrodinger operators $L_i, i\in \Z$ forms
  a $(2n+1)$-cyclic factorization chain with shift $\Delta\in \C$ if
  in addition to \eqref{eq:Dxform} we also have
  \begin{equation}\label{eq:shift} 
    U_{i+2n+1} = U_i+\Delta,\quad i\in \Z.
  \end{equation}
\end{definition}
\begin{prop}\label{prop:dchainfchain}
  Suppose that  the Schr\"odinger operators $L_i,\; i\in \Z$ form a
  factorization chain with shift $\Delta$. Then, the
  corresponding $w_0,\ldots, w_{2n}$ and
  \begin{equation}
    \label{eq:alphaidef}
    a_i =  \lambda_{i} - \lambda_{i+1},\quad i=0,\ldots, 2n
  \end{equation}
  form a $(2n+1)$-cyclic Darboux dressing chain with shift $\Delta$.
\end{prop}

\begin{proof}
  Eliminating the potentials in \eqref{eq:Riccati} and using
  \eqref{eq:alphaidef} to define $a_i$, we obtain the system of coupled equations
  \[
    (w_i + w_{i+1})' + w_{i+1}^2 - w_i^2 = a_i ,\quad i\in \N
  \]
  whose form coincides with \eqref{eq:wfchain}.  Relation
  \eqref{eq:shift} implies that $w_{i+2n+1} = w_i$ and that
  $\lambda_{i+2n+1} = \lambda_i + \Delta$.  The latter implies that
  $a_{i+2n+1} = a_i$ also. Hence the infinite chain of equations
  relating $w_i, w_{i+1}, a_i$ closes onto the finite system
  \eqref{eq:wfchain}. Since $\lambda_{2n+1}=\lambda_0+\Delta$, from
  \eqref{eq:alphaidef} it follows that \eqref{eq:Deltasumalpha} holds.
\end{proof}

\section{Maya diagrams and trivial monodromy potentials}\label{sec:Maya}

In this Section we introduce the main elements and results needed for
the classification of rational solutions of odd-cyclic dressing
chains. In Section~\ref{sec:veselov} we will prove the main
characterization result, namely that all rational solutions of an
odd-cyclic dressing chain can be expressed as log-derivatives of
Wronskian determinants whose entries are Hermite polynomials. The
basis for this proof lies in the theory of Schr\"odinger operators
with trivial monodromy, for which we refer to the celebrated papers of
Duistermaat and Gr\"unbaum \cite{refDG86} and Oblomkov
\cite{oblomkov1999monodromy}. However, before we can state the main
theorem we need to recall some basic definitions on Maya diagrams and
Hermite pseudo-Wronskians, which will be the building blocks of all
solutions.

\subsection{Maya diagrams}
Following Noumi \cite{noumi2004painleve}, we define a Maya diagram in the following manner.
\begin{definition}
  A Maya diagram is a set of integers $M\subset\Z$ that contains a
  finite number of positive integers, and excludes a finite number of
  negative integers. 
\end{definition}

\begin{definition}\label{def:index}
  Let $M\subset\Z$ be a Maya diagram.  Let $m_j,\; j\in \N$ be the
    $j\supth$ largest element of $M$ so that $m_1>m_2>\cdots$ is a
    decreasing enumeration of $M$.  Since only finitely many negative
    numbers are missing from $M$, we must have $m_{j+1} = m_j+1$ for
    $j$ sufficiently large.  Hence, there exists an $s_M\in \Z$ such
    that $m_j = -j+s_M$ for all $j$ sufficiently large. We call $s_M$
    the index of $M$.  Alternatively, the index of $M$ can also be defined in terms of the corresponding partition (see  \cite[Section 5.1]{noumi2004painleve}). 
\end{definition}
 
A Maya diagram can be visually represented as a sequence of $\boxdot$
and $\emptybox$ symbols with the filled symbol $\boxdot$ in position
$i$ indicating membership $i\in M$. A Maya diagram thus begins with an
infinite sequence of filled $\boxdot$ and terminates with an infinite
sequence of empty $\emptybox$.

We next describe the various forms to label Maya diagrams.

\begin{definition}

Let $M$ be a Maya diagram, and 
\[ M_-= \{ -m-1 \colon m\notin M, m<0\},\qquad M_+ = \{ m\colon m\in
M\,, m\geq 0 \}. \]
Let $s_1>s_2>\cdots > s_p$ and $t_1> t_2>\dots> t_q$ be the
elements of $M_-$ and $M_+$ arranged in descending order. 
The  \textit{Frobenius symbol} of $M$ is defined as the double list $(s_1,\ldots, s_p \mid t_q,\ldots, t_1)$, 
\end{definition}

If a Maya diagram $M$ has the Frobenius symbol $(s_1,\ldots, s_p \mid t_q,\ldots, t_1)$ , its index is given by $s_M=q-p$. 
The classical Frobenius symbol
\cite{andrews2004integer,olsson1994combinatorics,andrews1998theory} corresponds to the zero index case where $q=p$.

A natural operation in Maya diagrams is the following translation by an integer $k$
\begin{equation}\label{eq:trans}
M+k = \{ m+k \colon m\in M \},\qquad k\in \Z.
\end{equation}
The behaviour of the index $s_M$ under translation of $k$ is given by
\begin{equation}\label{eq:indexshift}
M'=M+k\quad \Rightarrow \quad s_{M'}=s_M+k.
\end{equation}
A Maya diagram $M\subset \Z$ is said to be in standard form if $p=0$
and $t_q>0$.  We visually recognize a Maya diagram in standard form when all the boxes to the left of the origin are filled $\boxdot$ and the first box to the right of the origin is empty
$\emptybox$.

Through this paper, we will adopt the convention to use Hermite polynomials $H_n(z),\; n=0,1,\ldots$ as univariate
polynomials defined by
\begin{equation}
  \label{eq:HnRodrigues}
  H_n(z)=(-1)^ne^{z^2}\frac{d^n}{dz^n}e^{-z^2},\qquad n=0,1,2,\dots.
\end{equation}
The $H_n(z)$ are known as classical orthogonal polynomials because they
satisfy a second-order eigenvalue equation
\begin{equation}
  \label{eq:hermde}
  y''-2zy' = -2n y,\quad y= H_n(z),
\end{equation}
and a 3-term recurrence relation
\begin{equation}
  \label{eq:h3term}
  H_{n+1}(z) =  2x H_n(z) - 2n H_{n-1}(z).
\end{equation}
The above properties imply the following orthogonality relation:
\begin{equation}
  \label{eq:hortho}
  \int_{\R} H_m(z) H_n(z) e^{-z^2} dz = \sqrt{\pi}\, 2^n
  n! \delta_{n,m},
\end{equation}
and generating function:
\begin{align}
  \label{eq:hermgf}
    e^{zt - \tfrac14 t^2}
    &= \sum_{n=0}^\infty H_n(z) \frac{t^n}{2^nn!}.
\end{align}
We will also make use of the {\it conjugate Hermite polynomials} defined by
\begin{equation}
    \label{eq:thndef}
    \th_n(z)={\rm i}^{-n} \h_{n}({\rm i}z).
  \end{equation}
These conjugate polynomials can also be defined by the relation
\begin{equation}
    \th_n(z)=e^{-z^2}\frac{d^n}{dz^n}e^{z^2},\qquad n=0,1,2,\dots.
\end{equation}
and they satisfy the differential equation
\begin{equation}
    y''-2zy' = 2(n+1) y,\quad y= e^{z^2}\th_n(z),\; n=0,1,2,\ldots.
\end{equation}

 Following \cite{gomez2016durfee}, to every Maya diagram we
 associate a polynomial called a Hermite pseudo-Wronskian.
\begin{definition}
  Let $M$ be a Maya diagram and $(s_1,\dots,s_r|t_q,\dots,t_1)$ its
  corresponding Frobenius symbol. We define the polynomial
  \begin{equation}\label{eq:pWdef1} \H_M(z) = \ex{-rz^2}\Wr[ \ex{z^2}
    \th_{s_1},\ldots, \ex{z^2} \th_{s_r}, \h_{t_q},\ldots \h_{t_1} ],
  \end{equation} where $\Wr$ denotes the Wronskian determinant of the indicated functions.
  
The polynomial nature of $\H_M(z)$ becomes evident in the following determinantal representation (see Proposition 7 in \cite{gomez2016durfee})
\begin{equation}\label{eq:pWdef2} \H_M(z) =
    \begin{vmatrix} \th_{s_1} & \th_{s_1+1} & \ldots &
\th_{s_1+r+q-1}\\ \vdots & \vdots & \ddots & \vdots\\ \th_{s_r} &
\th_{s_r+1} & \ldots & \th_{s_r+r+q-1}\\ \h_{t_q} & D_{z} \h_{t_q} &
\ldots & D_{z}^{r+q-1}\h_{t_q}\\ \vdots & \vdots & \ddots & \vdots\\
\h_{t_1} & D_{z} \h_{t_1} & \ldots & D_{z}^{r+q-1}\h_{t_1}
    \end{vmatrix}
  \end{equation}
\end{definition}

\begin{prop}[Theorem 1 in \cite{gomez2016durfee}]\label{prop:Mshift}
For any $k\in\Z$, the Hermite pseudo-Wronskians $H_M$ and $H_{M+k}$ coincide up to a multiplicative constant.
\end{prop}
In fact, with the following suitable rescaling of $H_M$
 \begin{equation}
    \label{eq:hHdef}
    \widehat{\H}_M(z) = c_M \H_M(z),\qquad c_M=\frac{(-1)^{rq}}{\prod_{1\leq i<j\leq r} (2s_j-2s_i)\prod_{1\leq
        i<j\leq q}
      (2 t_i-2t_j)},
  \end{equation}
we have 
  \begin{equation} \label{eq:HMequiv}
       \widehat{\H}_M(z) =  \widehat{\H}_{M+k}(z),\qquad \forall k \in \Z.
  \end{equation}

This property becomes useful if we observe that in every equivalence class of Maya diagrams related by translations, there is a unique representative where $M$ is in standard form, whose associated $H_M$ is a pure Wronskian determinant of Hermite polynomials. Although we could restrict the analysis without loss of generality to Maya diagrams in standard form and Wronskians of Hermite polynomials,  we will employ the general notation as it brings conceptual clarity to the description of Maya cycles in Section~\ref{sec:cycles} .

To every Maya diagram $M$ we can associate a potential with quadratic growth at infinity in the following manner
\begin{definition}
We define a \textit{rational extension of the harmonic oscillator} as the Schr\"odinger operator 
\begin{eqnarray}
  L_M&=&-D_{zz}+U_M(z) ,\label{eq:LUM}\\
  U_M(z)& =& z^2 - 2 D_z^2 \log \H_M(z) + 2s_M,  \label{eq:UMdef}
\end{eqnarray}
where $\H_M(z)$ is the associated pseudo-Wronskian \eqref{eq:pWdef1}--\eqref{eq:pWdef2}, and $s_M\in \Z$ is the index of $M$.
\end{definition}

The functions $U_M$ are called rational extensions of the harmonic oscillator because they contain the harmonic term plus a rational term that vanishes for large $z$.  Krein
  derived a general condition \cite{krein57} for regularity of
  potentials constructed via multi-step Darboux transformations, although
  Krein's brief announcement was limited to potentials on  the half-line
  and it did not consider exactly solvable examples.  Adler presented a
  full proof of this regularity result (no poles on the real line) using
  disconjugacy techniques  \cite{adler1994modification}, and  considered the case of the
  harmonic oscillator as a particular example.  The proof of the
  result for potentials defined on all of $\R$ follows from a more
  general theorem that counts the number of real poles of such
  potentials \cite{garcia2015oscillation}.


Let us also note that the above pseudo-Wronskians are related to the
$\tau$-functions of the KP hierarchy and that the factorization chain
is equivalent to a chain of Hirota bilinear relations.  For more
details, see the chapter by the present authors in \cite{crc2019}.

\subsection{Trivial monodromy}

\begin{definition}\label{def:ratext}
  A Schr\"odinger operator $L=-D_{zz}+U(z)$ has trivial monodromy at $\xi\in\C$ if the general solution of the equation
  \[L[\psi]=-\psi''+U\psi = \lambda\psi\]
  is meromorphic in a neighbourhood of $\xi$ for all values of $\lambda\in\C$.
  If $L$ has trivial monodromy at every point $\xi\in\C$ we say that $L$ is monodromy-free.
\end{definition}

Duistermaat and Gr\"unbaum proved that the condition that $L$ has trivial monodromy at $\xi\in\C$ is equivalent to certain restrictions on the coefficients of the Laurent series expansion of the potential.

\begin{prop}[Proposition 3.3 in \cite{refDG86}]
  \label{prop:dg}
  Let $U(z)$ be meromorphic in a neighbourhood of $z=\xi$ with Laurent
  expansion
  \[ U(z) = \sum_{j\geq -2} c_j (z-\xi)^j,\quad c_{-2}\neq 0.\] 
Then the Schr\"odinger operator $L=-D_{zz} + U(z)$ has trivial monodromy
  at $z=\xi$ if and only if there exists an integer $\nu \geq 1$ such
  that
  \begin{equation}
    \label{eq:ccond}
    c_{-2} = \nu (\nu+1),\qquad c_{2j-1} = 0,\quad 0\leq j \leq
    \nu. 
  \end{equation}

\end{prop}

Oblomkov classified monodromy-free potentials with quadratic growth at
infinity, finding that they can all be obtained by a finite sequence
of rational Darboux transformations applied on the harmonic
oscillator.

\begin{prop}[Theorem 3 in \cite{oblomkov1999monodromy}]\label{prop:Obl}
  The rational extensions of the harmonic oscillator  $L_M$  \eqref{eq:LUM} \eqref{eq:UMdef}
  have trivial monodromy. Conversely, if a  Schr\"odinger operator $L=-D_{zz}+U(z)$ has trivial
  monodromy and the potential has quadratic growth at infinity, then, up to an additive constant, $U=U_M$ for some Maya
  diagram $M$.
\end{prop}

We see thus that the class of monodromy-free potentials with quadratic growth at infinity coincides with the class of rational extensions of the harmonic oscillator given in Definition~\ref{def:ratext}. To every Maya diagram $M$ there corresponds a monodromy-free Schr\"odinger operator whose potential is a rational extension of the harmonic oscillator. The set of rational Darboux transformations preserves this class of operators. More specifically, a single step Darboux transformation \eqref{eq:Dxform} on a Schr\"odinger operator of the form \eqref{eq:LUM}-\eqref{eq:UMdef} leads to another rational extension whose Maya diagram differs from the previous one by a single flip.

\begin{definition}
Given a Maya diagram $M$, we define the flip at position $m\in \Z$ to be the involution
\begin{equation}\label{eq:flipdef}
 \phi_m : M \mapsto
\begin{cases}
   M \cup \{ m \}, & \text{if}\quad m\notin M, \\
   M \setminus \{ m \},\quad & \text{if}\quad m\in M.
\end{cases}
\end{equation}
\end{definition}
\noindent
In the first case, we say that $\phi_m$ acts on $M$ by a
state-deleting transformation ($\emptybox\to \boxdot$).  In the second
case, we say that $\phi_m$ acts by a state-adding transformation
($\boxdot\to\emptybox$).

\begin{prop}[Proposition 3.11 \cite{clarkson2020cyclic}]
Two Maya diagrams $M, M'$ are related by a flip \eqref{eq:flipdef}
  if and only if their associated rational extensions $U_M,U_{M'}$, see \eqref{eq:UMdef}, are
  connected by a Darboux transformation \eqref{Ui-Ui+1}.
\end{prop}

Exceptional orthogonal poynomials are intimately related with Darboux transformations of Schr\"odinger operators, \cite{gomez2013conjecture,garcia2016bochner,garcia2021sigma}. In fact, the bound states of  operators \eqref{eq:LUM}-\eqref{eq:UMdef} essentially define exceptional Hermite polynomials, \cite{gomez2013rational,gomez2019ladder,gomez2019corrigendum}.
More generally, it will be useful to characterize the class of quasi-rational eigenfunctions of \eqref{eq:LUM}-\eqref{eq:UMdef}. We recall that $f(z)$ is \textit{quasi-rational }if $(\log f)'$ is a rational function of $z$.

\begin{prop}\label{prop:seedfunc}
  Up to a scalar multiple, every quasi-rational eigenfunction of the
  Schr\"odinger operator $\cL_M=-D_{zz}+U_M(z)$ with $U_M$ as in
  \eqref{eq:UMdef} has the form
  \begin{equation}
    \label{eq:seedfunc}
    \psi_{M,m} = \exp(\tfrac12\sigma z^2)\frac{\H_{\phi_m(M)}(z)}{\H_M(z)}, \qquad m\in \Z,
  \end{equation}
  with 
  \[ \sigma = \begin{cases}
      -1, & \text{if}\quad m\notin M, \\
      +1,\quad & \text{if}\quad m\in M,
    \end{cases} 
  \]

  Explicitly, we have
  \begin{equation}
    \label{eq:Mneigenfunc}
    \cL_M \psi_{M,m}  = (2m+1) \psi_{M,m} ,\quad m\in \Z.    
  \end{equation}
\end{prop}
\begin{proof}
  We first prove that \eqref{eq:seedfunc} implies
  \eqref{eq:Mneigenfunc}. Consider a sequence of Maya diagrams
  $M_0,\ldots, M_n, M_{n+1}$ where $M_0 = \Z_-$ is the trivial Maya
  diagram,
  \[  M_n=M,\qquad  M_{n+1} = \phi_m(M),\]
   and each Maya diagram
  \[ M_{i+1}=\phi_{\mu_i}(M_i),\qquad  i=0,\ldots, n-1, \] differs from the
  preceeding one $M_i$ by a single flip at position $\mu_i\in \Z$.
  Let
  \[ \sigma_i=
    \begin{cases}
      -1,& \text{ if }\mu_i\notin M_i\\
      +1, & \text{ otherwise}.
    \end{cases}\]
  The unique quasi-rational eigenfunctions  \cite{erdelyiv1} of the
  classical harmonic oscillator
  operator $L_0 = -D_{zz} + z^2$ are
  \[ L_0 \psi_m = (2m+1) \psi_m,\]
  where
  \[
    \psi_m(z) =
    \begin{cases}
      H_m(z) e^{-z^2/2} & \text{ if } m\ge 0,\\
      \th_{-m}(z) e^{z^2/2} & \text{ if } m<0.
    \end{cases}
  \]
  A straightforward induction shows that
  $L_i = L_{M_i},\; i=0,\ldots, n+1$ is a factorization chain with the
  corresponding
  \begin{align}\label{eq:wiproof}
    w_i(z)
    &= \sigma_{i} z+ \frac{\H_{M_{i+1}}'(z)}{\H_{M_{i+1}}(z)}-
      \frac{\H_{M_{i}}'(z)}{\H_{M_{i}}(z)},\\
    \lambda_i &= 2\mu_i+1.\label{eq:liproof}
  \end{align}
Since $w_{n}$ is the log-derivative of $\psi_{M,m}$, the eigenvalue
  relation \eqref{eq:Mneigenfunc} follows immediately.

  Conversely, suppose that
  \[ L_M \hat\psi = \hat\lambda \hat\psi,\] and that
  $\hw(z)=\log(\hat\psi(z))'$ is a rational function.  As above, let
  $M_0,\ldots, M_n=M$ be a sequence of Maya diagrams such that $M_0$
  is trivial and 
  $M_{i+1} = \phi_{\mu_i}(M_i),\; i=0,\ldots, n-1$.  Let
  $L_0,\ldots, L_n = L_M$ be the corresponding factorization chain of Schr\"odinger operators and $(w_i|a_i)_{i=0}^{n-1}$ as in \eqref{eq:wiproof}-\eqref{eq:liproof} with $a_i=\lambda_{i+1}-\lambda_i$ be
  the corresponding dressing chain.  We can extend the chain by setting
  $w_n=\hw$ and $a_n= \hat\lambda-\lambda_n$.  Let
  \begin{equation}
    \label{eq:hwn-1}
     \hw_{n-1} = w_{n-1} + \frac{a_{n-1}}{w_{n-1}+ w_n}, \quad
    \ha_{n-2} = a_{n-2}+ a_{n-1},
  \end{equation}
  and observe that
  \begin{align*}
    w_{n-2}'&+\hw_{n-1}' + \hw_{n-1}^2- w_{n-2}^2\\
    &=
      w_{n-2}'-\frac{a_{n-1}}{w_{n-1}+w_n}\frac{(w_{n-1}+w_n)'}{w_{n-1}+w_n}
       + \lp w_{n-1} + \frac{a_{n-1}}{w_{n-1}+ w_n}\rp^2- w_{n-2}^2\\
    &=
      w_{n-2}'-\frac{a_{n-1}}{w_{n-1}+w_n}\lp w_{n-1}- w_n +
      \frac{a_{n-1}}{w_{n-1}+w_n}\rp 
       + \lp w_{n-1} + \frac{a_{n-1}}{w_{n-1}+ w_n}\rp^2- w_{n-2}^2\\
    &=a_{n-1}+ a_{n-2} = \ha_{n-2}.
  \end{align*}
  In this way we obtain a shorter dressing chain
  $w_0,\ldots,w_{n-2},\hw_{n-1}$ where all of the components are
  rational functions.  Continuing this argument inductively we arrive at a
  rational function $\tw(z)$ that satisfies the Ricatti equation
  \[ \tw' + \tw^2 = z^2- \tilde\lambda,\]
  and is related to $\hw$ by a sequence of rational
  transformations \eqref{eq:hwn-1}.  We conclude that
  $\tw(z)$ is the log-derivative of a quasi-rational eigenfunction of
  the classical harmonic oscillator $L_0$, and hence either
  \[ \tw = -z + \frac{H_{m}'}{H_{m}},\quad m\ge 0 \]
  or
  \[ \tw = z  + \frac{\th_{-m}'}{\th_{-m}},\quad m< 0. \]
  Successively applying the inverse of the rational transformation
  \eqref{eq:hwn-1} we conclude that
  \[ \hw = \pm z + \frac{H_{\phi_m(M)}'}{H_{\phi_m(M)}} -
    \frac{H_M'}{H_M}.\]
  Therefore, up to a non-zero scalar multiple the corresponding
  $\hat\psi$ must have the form \eqref{eq:seedfunc}.  
\end{proof}

Note that this characterization covers all quasi-rational
eigenfunctions, not just the square integrable ones. For our purpose
of classifying rational solutions to $(2n+1)$-cyclic dressing chains
this is the relevant class, and square integrability of the
eigenfunctions plays no role. Therefore, we employ the term
\textit{eigenfunction} in this formal sense, as solutions to the
eigenvalue problem.

\section{Characterization of rational solutions to odd-cyclic dressing chains }\label{sec:veselov}

In this section we state and prove the main result that allows the
classification of rational solutions to the $A_{2n}$-\p\ system,
namely that all of them belong to the class of rational extensions of
the harmonic oscillator. Most contents of this Section follow closely
the results obtained by Veselov in \cite{veselov2001stieltjes},
adapting the notation to our needs and providing further proofs for
intermediate results where we found it necessary.

We start by proving that the only possible poles of $w_i$ are simple, and growth at infinity is at most linear.

\begin{prop}\label{prop:wform}

 If $(w_0,\dots,w_{2n}|a_0,\dots,a_{2n})$  is a rational solution of  a $(2n+1)$-cyclic  dressing chain, then each function $w_i$  necessarily has the form
  \begin{equation}
    \label{eq:wiform}
    w_i = \pm A z + B_i + \sum_{j=1}^N \frac{A_{ij}}{z-\zeta_j},\quad
    A,B_i,A_{ij},\zeta_j\in \C,\; A\neq 0 \qquad i=0,1,\ldots,2n.  
  \end{equation}
\end{prop}
\begin{proof}
  We can rewrite the dressing chain equations \eqref{eq:wfchain} as
  \[ f_i' + f_i d_i = \alpha_i,\qquad f_i' \neq 0,\]
  where 
  \[ f_i(z) = c(w_i+w_{i+1})(cz),\quad d_i(z) = c(w_i-w_{i+1})(cz),\quad c^2=-1/\Delta.\]
  Hence, we can write
  \begin{align}
    \label{eq:wi1}
      2cw_i(cz) &= f_i(z)  -\frac{f_i'(z)}{f_i(z)} + \frac{\alpha_i}{f_i(z)}\\
      \label{eq:wi2}
      2cw_{i+1}(cz) &= f_i(z)  +\frac{f_i'(z)}{f_i(z)} - \frac{\alpha_i}{f_i(z)}.
  \end{align}
  Suppose that each $w_i$ in the chain has the following behaviour for
  large $z$
  \[ c w_i(cz) = A_i z^{k_i} + O(z^{k_i-1}),\quad z\to\infty,\qquad A_i\neq
    0,\quad i=0,1,\ldots, n.\]

  Our first claim is that $k_i \geq 0$ for all $i$.  This follows by
  inspection of \eqref{eq:wi1}.  Our second claim is that that
  $k_{i+1} = k_i$ for all $i$.  If $k_i=0$ for every $i$, then the
  claim follows trivially.  Suppose then that $k_i >0$ for at least
  one $i$.  By \eqref{eq:wi1}, it is clear that as $ z\to\infty $
  either
  \begin{align*}
    f_i &= 2 A_i z^{k_i} + O(z^{k_i-1}),\\   \intertext{or}
   \frac{\alpha_i}{f_i}&= 2 A_i z^{k_i} + O(z^{k_i-1}),\quad 
                          \alpha_i\neq 0.
  \end{align*}
  In any of the two cases we have
  \begin{equation}
    \label{eq:w+1pmai}
     c w_{i+1}(cz) = \pm A_i z^{k_i} + O(z^{k_i-1}), \quad z\to\infty,
  \end{equation}
  thereby proving the claim.  Our third claim is that $k_i=1$ for all
  $i$. This follows because the sum in \eqref{eq:wsum} involves an odd
  number of terms.  Finally, we conclude that $A_{i+1} = \pm A_i$ by
  \eqref{eq:w+1pmai}.

  Let $\zeta_j,\; j=1,\ldots N$ be the poles of $w_1,\ldots, w_n$. Fix
  a $j$ and write the Laurent expansion
  \[c w_i(cz) = A_{ij} (z-\zeta_j)^{-\ell_i} +
    O((z-\zeta_j)^{-\ell_i+1}),\quad z\to \zeta_j,\qquad A_{ij}\neq
    0,\quad i=0,1,\ldots,n,\] We claim that $\ell_i = 1$ for all
  $i$.  Suppose not and that $\ell_i \geq 2$ for some $i$.  By
  \eqref{eq:wi1}, as $ z\to\zeta_j $ either
  \begin{align*}
    f_i &= 2 A_{ij} (z-\zeta_j)^{-\ell_i} +
          O((z-\zeta_j)^{-\ell_i+1}),\\   \intertext{or} 
    \frac{\alpha_i}{f_i}&= 2 A_{ij} (z-\zeta_j)^{-\ell_i} +
                          O((z-\zeta_j)^{-\ell_i+1}),\quad   
                          \alpha_i\neq 0.
  \end{align*}
  In both cases,
  \[ c w_{i+1}(cz) = \pm A_{ij} (z-\zeta_j)^{-\ell_i} +
    O((z-\zeta_j)^{-\ell_i+1}), \quad  
    z\to\zeta_j.\] Thus, $\ell_{i+1}= \ell_i$ and $A_{i+1,j} = \pm A_{i,j}$
  for every $i$.  Since the sum in \eqref{eq:wsum} involves an odd
  number of terms, this leads to a contradiction.
\end{proof}

\begin{prop}
  \label{prop:res2}
  Let $(w_0,\dots,w_{2n}|a_0,\dots,a_{2n})$  be a rational solution of a $(2n+1)$-cyclic dressing chain
  and let  $\zeta\in \C$ be a pole of some function $w_i$ in the chain. Then we have
  \begin{equation}
  \Res_\zeta w_i^2 =0,\qquad \Res_\zeta w_i\in\Z, \qquad |\Res_\zeta w_i|\leq n.
  \end{equation}
\end{prop}

  We need to show that if $\zeta \in \C$ is a pole of $w_i$, then
  \begin{equation}
    \label{eq:wimizeta}
    w_i = m_i (z-\zeta)^{-1} + O(z-\zeta),\quad z\to \zeta,\quad m_i
    \in \Z.
  \end{equation}
  By Proposition~\ref{prop:wform}, $\zeta$ is a simple pole of $w_i$ so the local behaviour of $w_i$ near $\zeta$ is
  \begin{equation}
    \label{eq:wiaiz-1}
    w_i = A_i (z-\zeta) ^{-1} + B_i + O(z-\zeta),\quad z\to \zeta,\quad
    i=0,1,\ldots,2n.
  \end{equation}
Inserting the above expansion for $w_i$ in the equations of the dressing chain \eqref{eq:wfchain} and collecting the leading order terms at  $(z-\zeta)^{-2}$ and $(z-\zeta)^{-1}$ we obtain the relations
  \begin{eqnarray}
    \label{eq:aiai1}
     -(A_i +A_{i+1}) + A_{i+1}^2 - A_i^2 &= &0 ,\\
     \label{eq:aibi}
     A_{i+1} B_{i+1}  - A_i B_i &=& 0,\qquad i=0,\dots,2n \mod(2n+1).
  \end{eqnarray}  
 These equations, together with the constraints on $\{A_i,B_i\}, \quad i=0,\dots,2n$ derived from the closure condition \eqref{eq:wsum} are enough to prove the desired claim, which proceeds by deriving three chained lemmas.
  
\begin{lem}
  \label{lem:a1}
   Let $\zeta$ be a simple pole of a rational function  in a $(2n+1)$-cyclic dressing chain as per \eqref{eq:wiaiz-1}, and let $\{A_i\}_{i=0}^{2n}$ be the sequence of residues of $w_i$ at $\zeta$ as per \eqref{eq:wiaiz-1}  . For each $i=1,2,\ldots,2n+1$,  there exists a $k_i \in \{ 1,\ldots, 2i \}$
  such that
  \begin{equation}
    \label{eq:aia0ki}
     A_i = (-1)^{k_i} A_0 + k_i-i.
  \end{equation}
\end{lem}
\begin{proof}
  The proof is by induction on $i$.  By \eqref{eq:aiai1} we have
  \begin{equation}
    \label{eq:ai+1ai}
     A_{i+1} = -A_i,\quad \text{or} \quad A_{i+1} = A_i +1 .   
  \end{equation}
  For $i=0$, the first case corresponds to $k_1 = 1$, and second case
  to $k_1=2$.  Suppose \eqref{eq:aia0ki} holds for a given $i$.
  Hence,
  \[ A_{i+1} = (-1)^{k_i+1} A_0 -k_i + i,\quad \text{or} \quad
  A_{i+1} = (-1)^{k_i} A_0 + k_i-i+1.\]
  The first possibility corresponds to
  \[ k_{i+1} =  - k_i +2 i+1, \]
  and the second possibility corresponds to 
  \[ k_{i+1} = k_i+2,\]
but in both cases $k_{i+1}\in\{1,\dots,2(i+1)\}$, thus establishing the claim.
\end{proof}
\begin{lem}
  \label{lem:a2}
  Let $\zeta$ be a simple pole of a rational function $w_i$ that solves a $(2n+1)$-cyclic dressing chain. Then the residue $A_i$ at $\zeta$ must be an integer $A_i\in\{-n,\dots,n\}$.
\end{lem}
\begin{proof}
The results follows trivially from the previous lemma and the closure condition. Indeed, Lemma~\ref{lem:a2} for $i=2n+1$ reads
 \[ A_{2n+1} =A_0= (-1)^{k_{2n+1}} A_0 + k_{2n+1} -(2n+1) .\]
 The closure condition $A_{2n+1}=A_0$ implies that the second possibility in \eqref{eq:ai+1ai} occurs an even number of times and the first possibility an odd number of times. Since $k_0=0$, we see that  $k_{2n+1}$ must be an odd number, and therefore we can write $k_{2n+1}=2j+1$ with $j\in\{0,\dots,2n\}$.  Hence,
  \[ 2 A_0 = k_{2n+1}-(2n+1) = 2 (j-n),\] which proves the claim for the residue $A_0$.
The result extends from $A_0$ to any $A_i$ in the chain by cyclicity.
\end{proof}
\begin{lem}
  \label{lem:a3}
  Let $\zeta$ be a simple pole of a rational function  in a $(2n+1)$-cyclic dressing chain as per \eqref{eq:wiaiz-1}, and let $\{A_i\}_{i=0}^{2n}$ be the sequence of residues of $w_i$ at $\zeta$ as per \eqref{eq:wiaiz-1} . Then  $A_i=0$ for some $i=0,1,\ldots, 2n$.
\end{lem}
\begin{proof}
  We argue by contradiction and suppose that  $\sgn A_i \in \{-1,1\}$ for all $i=0,1,\ldots,
  2n$.
  Hence, 
  \[ \sgn A_{i+1} = \pm \sgn A_i.\]
  From the cyclic condition  $A_{2n+1}=A_0$, it follows that in the set $i\in\{0,\dots,2n\}$ there must be an even number of indices such that
  $\sgn A_{i+1} = -\sgn A_i,$ and therefore an odd number of indices
  such that
  \[ |A_{i+1}| = |A_i|\pm 1 .\]
  It follows that $|A_{2n+1}|-|A_0|$ is an odd integer, which is a contradiction.
\end{proof}

The previous lemma implies that it is impossible that all the rational functions in the dressing chain have a common pole. We are now ready to conclude the proof of Proposition~\ref{prop:res2}.
\begin{proof}[Proof of Proposition \ref{prop:res2}]

  From Lemma \ref{lem:a3} it follows that $A_i B_i = 0$ for at least one
  $i=0,1,\ldots, 2n$.  Hence, by \eqref{eq:aibi}, $A_i B_i = 0$ for all
  $i=0,1,\ldots, 2n$; i.e. either $A_i=0$ or $B_i=0$ for every $i$.
  Relative to form \eqref{eq:wiaiz-1}, this is equivalent to the
  condition that $\Res_\zeta w_i^2=0$.  The integrality and bounds on the possible values of the residues $A_i$ follow from  Lemma \ref{lem:a2}.
\end{proof}

By  Proposition~\ref{prop:res2}, the expansion \eqref{eq:wimizeta} holds at every
pole $z=\zeta$ of a $(2n+1)$-cyclic factorization chain.  In particular
the residue $m_i=\Res_\zeta w_i$ is an integer and $|m_i|\leq n$.  The conclusions of
that proposition can be strengthened in the following manner.
\begin{prop}
  \label{prop:res2j}
  Let $(w_0,\dots,w_{2n}|a_0,\dots,a_{2n})$  be a rational solution of a $(2n+1)$-cyclic dressing chain.
  Let  $\zeta\in \C$ be a pole of a function $w_i$ in the chain and $m_i=\Res_\zeta w_i$. Then we have
  \begin{equation}
    \label{eq:res2j}
    \Res_\zeta w_i^{2j} = 0,\quad j=1,\ldots, |m_i|.
  \end{equation}
\end{prop}
In a similar manner, we structure the proof of this result in three simple lemmas. Consider the local expansion of $w_i$ around the pole at $z=\zeta$, which according to Proposition~\ref{prop:res2} has the form
\begin{equation}
\label{eq:wimibij}
 w_i = m_i (z-\zeta)^{-1} + \sum_{j=0}^\infty  B_{ij} (z-\zeta)^j,
\end{equation}
where $m_i\in \Z$ is an integer.

\begin{lem}
  \label{lem:a4}
Let  $S$ be the set of residues of all functions of the chain at $z=\zeta$:
\[ S= \{ m_i \colon i=0,1,\ldots, 2n \}.\] 
 Then, we have $-S=S$.
\end{lem}
\begin{proof}
  Consider an arbitrary $m_i\in S$. If $m_i=0$, then evidently
  $-m_i\in S$.  Suppose that $m_i>0$.  If $m_{i+1}=-m_i$, we are done.
  Otherwise let $k>0$ be the largest integer such that
  \[ |m_{i+1}| , \ldots, |m_{i+k}| > m_i.\]
  Such a $k$ must exist because of cyclicity.
  By \eqref{eq:ai+1ai},
  \[ |m_{i+k+1}| - |m_{i+k}| \in \{ -1,0,1\}.\]
  Since $k$ is as large as possible, $|m_{i+k+1}| = m_i$.  Suppose
  that $m_{i+k+1} = m_i$. That would mean from \eqref{eq:ai+1ai} that either $m_{i+k}=m_i-1$ or $m_{i+k}=-m_i$, both of which contradict the hypothesis that $|m_{i+k}| > m_i$.
  Therefore $m_{i+k+1} = -m_i$, and $-m_i\in S$ also.  The case $m_i<0$ is
  proved analogously.
\end{proof}

\begin{lem}
  \label{lem:a5}
  If $k\in S$ is positive, then $k-1\in S$ also.
\end{lem}
\begin{proof}
 We argue by contradiction and suppose that there exists a $k>0$ such that $k\in S$ but  $k-1\notin S$. By Lemma
  \ref{lem:a4}, $1-k\notin S$ also.  From \eqref{eq:ai+1ai} we have
  $|m_{i+1}| - |m_i| \in \{-1,0,1\}$ so it would follow that $|m_i| \geq k$ for all
  $i=0,\dots,2n,$ which contradicts Lemma \ref{lem:a3}.
\end{proof}

\begin{lem}
  \label{lem:a6}
  Let $m=\max\,\{m_i\}_{i=0}^{2n}$.  Then,
  \[ S = \{ -m, -m+1,\ldots, m-1,m \}.\]
\end{lem}

\begin{proof}
  This follows directly from     Lemmas \ref{lem:a4} and \ref{lem:a5}.
\end{proof}
We see that the set of residues at a given pole $z=\zeta$ along the chain contains all integer values between $-m$ and $m$, with $m\leq n$. We are now ready to prove Proposition~\ref{prop:res2j}.

\begin{proof}[Proof of Proposition \ref{prop:res2j}.]
  Given the expansion in $\eqref{eq:wimibij}$, the claim \eqref{eq:res2j} is equivalent to showing that  
  \[B_{i,2j-2} = 0,\text{ for all } j=1,\ldots, |m_i|.\]
   The argument proceeds  by induction
  on $j$.  Proposition \ref{prop:res2} established \eqref{eq:res2j}
  for $j=1$.  Suppose that \eqref{eq:res2j} holds for all $j\leq k$
  for a given $k\in \N$.  Suppose that $|m_i| > k$.  We will show that
  $B_{i,2k}=0$.

  Let $p< i$ be the largest integer such that $m_{p} = k$ and
  $q> i$ be the smallest integer such that $m_q = -k$.  If $k<|m_i|$, such $p,q$ are guaranteed to exist
  by Lemma \ref{lem:a6}.     Thus, by construction
  \begin{equation}\label{eq:mlcond}
    |m_\ell| \geq k+1,\quad \ell = p+1,\ldots ,i,\ldots, q-1.
    \end{equation}
  By the inductive assumption
  \[ B_{i,2j-2} = 0,\quad i=p,\ldots,q,\quad j=1,\ldots k.\]
  Hence, the vanishing of the coefficient of $z^{2k-1}$ in
  \eqref{eq:wfchain}, implies that
  \begin{equation} \label{eq:bl2k}
  (k+m_{\ell+1}) B_{\ell+1,2k} +(k-m_\ell) B_{\ell,2k} = 0,\quad   \ell=p,\ldots, q-1.
  \end{equation}
  Since $k+m_q=0$ and $k-m_p=0$, from \eqref{eq:mlcond} and \eqref{eq:bl2k} it follows that
  \[ B_{q-1,2k} = B_{p+1,2k} = 0.\]
 which in turn imply by cascade that 
  \[B_{\ell,2k} = 0 \text{ for all }\ell = p+1,\ldots, i  \ldots q-1,\]
  and in particular $B_{i,2k}=0$ as was to be shown.
\end{proof}


\begin{definition}\label{def:symdiff}
Given two sets $A$ and $B$ we define its \textit{symmetric difference} as the union of the set of elements of $A$ that are not in $B$ with the set of elements of $B$ that are not in $A$ 
  \begin{equation}
    \label{eq:ominusdef}
    A\ominus B = (A \setminus B) \cup (B\setminus A).
  \end{equation}
\end{definition}

The characterization of rational solutions of the $A_{2n}$ system can
be done in terms of cyclic Maya diagrams and Maya cycles --- concepts
that we introduce below.

\begin{definition}
  \label{def:multiset}
  Let $\hcZp,\; p\in \Nz$ denote set of integer sets of cardinality
  $p$, and let $\cZp$ denote the set of integer multisets of
  cardinality $p$.  We identify a set $\bbeta\in \hcZp$ with the
  strictly increasing integer sequence
  $\beta_0<\beta_1<\cdots <\beta_{p-1}$ that enumerates $\bbeta$ , and
  identify a multi-set $\bbeta\in \cZp$ with a non-decreasing integer
  sequence $\beta_0\le \beta_1 \le \cdots \le \beta_{p-1}$.  In this way, we regard $\hcZp$ as a subset of
  $\cZp$.
\end{definition}

We will use curly braces to denote both sets and multi-sets, letting the context resolve the ambiguity, and
round parentheses to denote sequences/tuples.

\begin{definition}
  \label{def:multiflip}
  For a set $\bbeta\in \hcZp$ let $\phi_{\bbeta}$ denote the
  \textit{multi-flip}
  \begin{equation}
    \label{eq:phimudef}
    \phi_{\bbeta}= \phi_{\beta_0} \circ \cdots \circ \phi_{\beta_{p-1}}.
  \end{equation}
  where the action of each single flip on a Maya diagram $M$ is given
  by \eqref{eq:flipdef}.  

\end{definition}

\noindent
We also use relation \eqref{eq:phimudef} to define $\phi_\bbe$ where
$\bbe\in \cZp$ is a multiset.  Since flips are involutions, we have
\[ \phi_\bbe = \phi_{\bbe'},\quad \bbe\in \cZ^p\] where $\bbe'$ is
the set consisting of the elements of $\bbe$ with an odd cardinality.

\begin{definition}\label{def:cyclicM}
  A Maya diagram $M$ is $p$-cyclic with shift $k$, or simply
  $(p,k)$-cyclic, if there exists a multi-flip
  $\phi_\bbe,\;\bbe \in \cZp$ such that
  \begin{equation}
    \label{eq:cyclicMdef}
    \phi_\bbe(M) = M+k.
  \end{equation}
  More generally, a Maya diagram $M$ will be said to be $p$-cyclic
  if it is $(p,k)$-cyclic for some shift $k\in \Z$.
\end{definition}

\begin{definition}
  \label{def:mayacycle}
  A $(p,k)$ \emph{Maya cycle} is a sequence of Maya diagrams
  $\bM=(M_0,M_1,\dots,M_{p})$ such that $M_i$ is
  related to $M_{i+1}$ by a single flip, and such that
  $M_p = M_0+k,\; k\in \Z$.  The sequence $\bmu \in \Z^p$ where
  \begin{equation}
    \label{eq:flipseq}
    \{\mu_i\} = M_{i+1} \ominus M_i,\quad i=0,\ldots, p-1 
  \end{equation}
  will be called the \emph{flip sequence}\footnote{Every element
      of a $(p,k)$ Maya cycle is a $(p,k)$ cyclic Maya diagram.  Thus,
      as a data structure, a Maya cycle can also be represented in
      terms of two components: (i) a cyclic Maya diagram $M_0$,
      together with (ii) a permutation of the flip sequence.  In the
      preliminary paper \cite{clarkson2020cyclic} we used  this
      representation.  However the present definition has certain advantages
      that become manifest once we consider the corresponding
      symmetries.}  of the Maya cycle, because, by construction,
  $\bmu$ is the unique sequence such that
  \[ M_{i+1} = \phi_{\mu_i}(M_i),\quad i=0,\ldots, p-1.\]
  For each $i=0,\ldots, p-1$ we also set
  \begin{equation}
    \label{eq:sign}
    \sigma_i=\begin{cases}
      -1, & \textit{if}\quad \mu_{i}\notin M_i, \\
      +1, & \textit{if}\quad \mu_{i}\in M_i,
  \end{cases}
  \end{equation}
  and refer to the sequence $\bsigma=(\sigma_0,\ldots, \sigma_{p-1})$
  as the \emph{sign sequence} of the Maya cycle.

  Observe that if
  $\bM$ is a Maya cycle, then so is
  \begin{equation}
    \label{eq:MCtrans}
    \bM+j = (M_0+j,M_1+j,\ldots, M_{p}+j),\quad j\in \Z.
  \end{equation}
  We use $\bM/\Z$ to denote the equivalence class of a Maya cycle
  $\bM$ modulo integer translations.
\end{definition}

We are now able to formulate the main theorem of this Section that characterizes rational solutions of an odd-cyclic dressing chain, and therefore  rational solutions of the $A_{2n}$ system.

\begin{thm}\label{thm:characterization}
  Let $\bM=(M_0,\ldots, M_{2n+1})$ be a $(2n+1,k)$ Maya cycle with flip sequence $\bmu=(\mu_0,\dots,\mu_{2n})\in \Z^{2n+1}$ and sign sequence $\bsigma=(\sigma_0,\dots,\sigma_{2n})\in \{-1,1\}^{2n+1}$. For
  $i=0,\ldots, 2n$, set 
  \begin{align}
     \label{eq:wsigMi}
    w_i(z)&= \sigma_{i} z+ \frac{\H_{M_{i+1}}'(z)}{\H_{M_{i+1}}(z)}-
            \frac{\H_{M_{i}}'(z)}{\H_{M_{i}}(z)},\\
    \label{eq:aimui}
    a_i&=2(\mu_{i}-\mu_{i+1}), \qquad \mu_{2n+1}=\mu_0+k,
  \end{align}
  where $ \H_{M_i}(z)$ and $ \H_{M_{i+1}}(z)$ are the corresponding Hermite
  pseudo-Wronskians \eqref{eq:pWdef2}.
  Then, $(\bw|\ba)=(w_0,\dots,w_{2n}|a_0,\dots,a_{2n})$ is a rational solution
  of a $(2n+1)$-cyclic dressing chain with shift $\Delta=2k$.  Conversely, every rational
  solution $(\bw|\ba)$ of a $(2n+1)$-cyclic dressing chain is determined in this fashion by a unique Maya cycle class $\bM/\Z$.
\end{thm}

\begin{proof}
  Let $\bM$ be a $(2n+1,k)$ Maya cycle and 
  \[ L_i=L_{M_i}=-D_{zz}+U_i(z),\qquad i=0,\ldots,2n+1,\] the
  corresponding sequence of rational extensions defined by
  \eqref{eq:LUM} \eqref{eq:UMdef}. The cyclicity condition $M_{2n+1}=M_0+k$, together with  \eqref{eq:indexshift}, \eqref{eq:HMequiv}, and  \eqref{eq:UMdef}  imply that
  \[ U_{2n+1}=U_0+2k, \]
  so the sequence $L_0,\dots,L_{2n+1}$ is factorization chain with shift $\Delta=2k$.
  Using definition \eqref{eq:seedfunc}, set
  \[ \psi_i = \psi_{M_i,\mu_i},\quad i=0,\ldots, 2n, \]
  so that
  \[ L_i \psi_i = (2\mu_i + 1)\psi_i. \] Let
  $w_i,a_i,\; i=0,\ldots, 2n$ be defined by \eqref{eq:wsigMi}
  \eqref{eq:aimui}.  Then, by Proposition \ref{prop:dchainfchain}
  and by Proposition \ref{prop:seedfunc}, the tuple
  $(w_0,\dots,w_{2n}|a_0,\dots,a_{2n})$ is a rational solution of a
  $(2n+1)$-cyclic dressing chain with shift $\Delta=2k$.
  
  We now prove the converse statement. We first show that the
  conditions satisfied by each rational solution $w_i$ at a pole
  $\zeta$, as expressed by Propositions~\ref{prop:res2} and
  \ref{prop:res2j} are precisely the conditions that express local
  trivial monodromy of the corresponding potential $U_i$.  Denote by
  $w$ any of the rational functions of a tuple
  $(w_0,\dots,w_{2n}|a_0,\dots,a_{2n})$ that satisfies a
  $(2n+1)$-cyclic dressing chain and let $\zeta$ be a pole of $w$. By
  Propositions~\ref{prop:res2} and \ref{prop:res2j}, the Laurent
  expansion of $w$ at $z=\zeta$ is
\begin{equation}\label{eq:wcond2}
w= \sum_{j=-1}^\infty b_j (z-\zeta)^j,\qquad b_{-1}=m\in\Z,\qquad b_{2j}=0 \text{ for } j=0,\dots,|m|-1.
\end{equation}
Since $U=w'+w^2+\lambda$ by \eqref{eq:Riccati}, the Laurent expansion at $\zeta$ of  $U$ is:
\[ U=   \sum_{j\geq -2} c_j (z-\xi)^j \]
where 
\[c_{-2}=  m(m-1),\qquad c_{-1}= 2mb_0,\qquad c_{0}= b_0^2 + (2m+1)b_0+\lambda\]
and
\begin{align}
c_{2j-1}&= 2jb_{2j} +2\sum_{i=-1}^{j-1} b_i b_{2j-i-1}, && j \geq 1 \label{eq:codd}\\
c_{2j}&= (2j+1)b_{2j+1} +b_j^2 + 2\sum_{i=-1}^{j-1} b_i b_{2j-i-1}, &&j \geq 1
\end{align}

By Proposition~\ref{prop:dg}, $U$ has trivial monodromy at $z=\zeta$ if and only if there exists an integer $\nu\geq 1$ such that \eqref{eq:ccond} holds. Depending on the sign of $m$, we choose $\nu$ as
\[
\nu=\begin{cases}
-m & \text{ if } m<0,\\
m-1 & \text{ if } m>0,\\
\end{cases}
\]
Note that if $w$ has pole at $z=\zeta$ with residue $m=-1$ or $m=0$, the potential $U$ is regular in a neighbourhood of $\zeta$. For other integer values of $m$, the conditions \eqref{eq:wcond2} on the even coefficients of $w$ imply that the precise number of odd coefficients \eqref{eq:codd} of $U$ vanish, as required by Proposition~\ref{prop:dg}. We conclude that $U$ has trivial monodromy at $z=\zeta$, and since $\zeta$ is arbitrary, $U$ is a monodromy-free potential. From Proposition~\ref{prop:wform} and  \eqref{eq:Riccati}  it follows that $U$ is a monodromy-free potential with quadratic growth at infinity, so Proposition~\ref{prop:Obl} implies that $U$ is a rational extension of the harmonic oscillator, i.e. it has the form \eqref{eq:UMdef} for some Maya diagram $M$. Recalling that $w$ is the log-derivative of the seed function for the Darboux transformation (see \eqref{eq:Lipsii}), and that all quasi-rational seed functions of potentials  \eqref{eq:UMdef} are characterized by Proposition~\ref{prop:seedfunc}, it suffices to take the log-derivative of \eqref{eq:seedfunc} to achieve the desired result \eqref{eq:wsigMi}. This argument was applied on an arbitrary element of the dressing chain, and therefore it applies to all such elements.

The cyclicity of the dressing chain implies that the
corresponding Schr\"odinger operators in the factorization chain given
by Proposition~\ref{prop:dchainfchain} are rational extensions of the
harmonic oscillator \eqref{eq:LUM}-\eqref{eq:UMdef}, and the closure
condition \eqref{eq:shift} defines a Maya cycle
$\bM=(M_0\dots,M_{2n},M_{2n+1})$ where $ M_{2n+1}=M_0+k$. From
Proposition~\ref{prop:MM+1} we see that all Maya cycles in the
equivalence class $\bM/\Z$ lead to the same rational solution.
\end{proof}

\section{Classification of Cyclic Maya diagrams}\label{sec:cycles}

In Section~\ref{sec:veselov} we saw that a rational solution to a
$(2n+1)$-cyclic dressing chain corresponds to a factorization chain
\eqref{eq:shift} whose potentials are rational extensions of the
harmonic oscillator and, as a consequence, they are indexed by Maya
diagrams. More specifically, components of the solution have the form
\eqref{eq:wsigMi} and they are essentially determined by two Maya
diagrams connected by a flip operation (or equivalently, by two
potentials related by a Darboux transformation).  The periodicity of
the dressing chain thus translates into a cyclicity condition on the
Maya diagrams. A sequence of flip operations encodes multi-step
Darboux transformations (also called Crum transformation
\cite{crum1955sturm}) at the level of Maya diagrams.

For a fixed shift $k$, every Maya diagram $M$ is $(p,k)$-cyclic for
some value of the period $p$ (Proposition 5.3 in
\cite{clarkson2020cyclic}). However, the relevant problem we need to
address is the converse: that of enumerating and classifying all
cyclic Maya diagrams for a fixed odd period $p=2n+1$.  This
  classification will then be extended in Section \ref{sect:colseq}
  to a classification of Maya cycles.

The desired classification for cyclic Maya diagrams of a fixed period
can be achieved by employing the key concepts of \textit{genus} and
\textit{interlacing}. The genus of a Maya diagram counts esentially
the number of blocks of filled boxes $\boxdot$ in the finite part of
$M$, and the initial and ending position of each block are called the
\textit{block coordinates}. Specifying its block coordinates
determines a Maya diagram uniquely, and this will be the most
convenient representation for our purpose of classifying cyclic Maya
diagrams. Let us make all these notions more precise.

\begin{definition}
  \label{def:genus} Let $\bbeta\in \cZ^{2\ell+1}$ be an integer multiset of
  odd cardinality with non-decreasing enumeration
  $\b_0\le\b_1\le\cdots\le \b_{2\ell}$.  Let $\Xi(\bbeta)$ be the Maya
  diagram defined by
    \begin{equation}
    \label{eq:Xidef} \Xi(\bbeta)= (-\infty,\beta_0) \cup [\beta_1,\beta_2)
    \cup \ \cdots \cup [\beta_{2\ell-1},\beta_{2\ell}),
  \end{equation}
  where
  \begin{equation} \label{eq:interval} [m,n) = \{ j\in \Z \colon m\leq
    j < n\}.
  \end{equation}
We refer to the set $\bbeta$ as the \textit{block coordinates} of the Maya diagram $M=\Xi(\bbeta)$.  
  
\end{definition}
It is important to note that if $\bbeta$ has repeated elements then the same Maya diagram admits a representation in terms of a smaller number of block coordinates. We make this notion precise in the following proposition.

\begin{prop}
  \label{prop:Xibiject}
  For every Maya diagram $M$, there exists a unique $g\in \Nz$ and a
  unique \emph{set} $\bbeta\in \hcZ^{2g+1}$ such that $M=\Xi(\bbeta)$.
\end{prop}
\begin{proof}
  Every Maya diagram $M$ has a unique description
  \[ M=(-\infty,\beta_0) \cup [\beta_1,\beta_2) \cup \ \cdots \cup
    [\beta_{2g-1},\beta_{2g}) \] where
  $\beta_0<\beta_1<\cdots< \beta_{2g}$ is an increasing integer
  sequence.  Observe that
  $\bbeta=\{ \beta_0,\beta_1,\ldots, \beta_{2g}\}$ is precisely the
  the set of integers that are in $M$ but are not in $M+1$ and viceversa.  Thus, the desired set $\bbeta \in \hcZ^{2g+1}$ can be given
  as
  \begin{equation}
    \label{eq:bbetaM+1}
     \bbeta = (M+1) \ominus M, 
  \end{equation}
  where $\ominus$ denotes the symmetric set difference defined in \eqref{eq:ominusdef}.
\end{proof}

  Given a multiset
$\bbeta'=\{\beta_0^{n_0},\dots,\beta_{p}^{n_p}\}\in\cZ^{2\ell+1}$ with
elements $\beta_i$ and multiplicities $n_i\in\N_0$, such that
$n_0+\cdots+n_p=2\ell+1$, the corresponding set of  block
  coordinates described by Proposition~\ref{prop:Xibiject} is given
by
\[ \bbeta= \{\beta_0^{m_0},\dots,\beta_{p}^{m_p}\}\in\hat\cZ^{2g+1},\quad \text{ where } m_i=n_i\!\!\!\!\mod(2),\quad i=0,\dots,p,  \]
and $m_0+\cdots+m_p=2g+1$. It is clear from \eqref{eq:Xidef} that $\Xi(\bbeta')=\Xi(\bbeta)$.

\begin{definition}
  Given a Maya diagram $M$, let $\bbeta$ be the set of cardinality
  $2g+1,\; g\in\Nz$ described in the above Proposition.  We say that $g$ is the \textit{genus} of $M$.
\end{definition}
\begin{prop}
  A genus $g$ Maya diagram is $(2g+1,1)$-cyclic.
\end{prop}
\begin{proof}
  As a direct  consequence of \eqref{eq:bbetaM+1} we have
  \[ \phi_{\bbeta}(M) = M+1.\]
\end{proof}

Let $\cZo$ denote the set of multisets of odd cardinality, and let
$\hcZo$ denote the set of sets of odd cardinality. The mapping
\eqref{eq:Xidef} defines a bijection $\Xi:\hcZo\to \cM$, which maps
cardinality to genus.  A multi-set $\bbeta\in \cZ^{(2g+1)}$
corresponds to a non-degreasing sequence
$\beta_0\le \beta_1\le \cdots\le \beta_{2g}$ and can also be used to
define a Maya diagram using \eqref{eq:Xidef}.  However, if $\bbeta$
has some repeated elements, then some of the blocks in
\eqref{eq:Xidef} will coalesce and result in a Maya diagram whose
genus is strictly smaller than $g$.  This observation may be
encapsulated by saying that the extended mapping $\Xi:\cZo\to M$ is
onto, but not one-to-one, and that the cardinality of the multiset
dominates the genus of the corresponding Maya diagram.

The visual explanation  of the  genus concept is clear in Figure~\ref{fig:genusM}. Removing the infinite initial 
  $\boxdot$ and  trailing $\emptybox$ segments, a
  Maya diagram consists of alternating empty $\emptybox$ and filled
  $\boxdot$ segments of finite variable length.  The genus $g$ counts the
  number of such pairs.  The  block coordinates $\beta_{2i}$ indicate the starting
  positions of the empty segments, and  $\beta_{2i+1}$ signal the starting positions of the filled
  segments.  Finally, note that $M$ is in standard form if and only if
  $\beta_0=0$.

With the well known corespondence between Maya diagrams and partitions, it is worth noting that the genus of a Maya diagram coincides with the number of distinct parts of the partition, \cite{clarkson2020cyclic}. This feature has been studied earlier in \cite{alladi1999fundamental}, in connection with some identities in the theory of $q$-series.

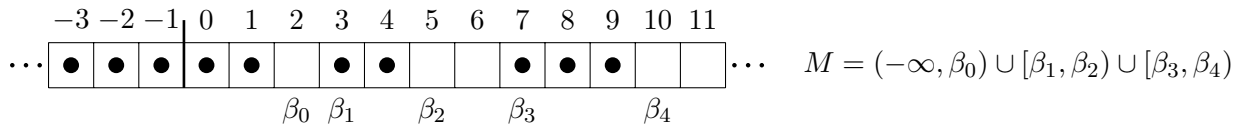
\begin{figure}[h]
\begin{tikzpicture}[scale=0.6]

\draw  (1,1) grid +(15 ,1);

\path [fill] (0.5,1.5) node {\huge ...} 
++(1,0) circle (5pt) ++(1,0) circle (5pt)  ++(1,0) circle (5pt) 
++(1,0) circle (5pt) ++(1,0) circle (5pt) 
++(2,0) circle (5pt) ++(1,0) circle (5pt) 
++ (3,0) circle (5pt)  ++(1,0) circle (5pt)   ++ (1,0) circle (5pt) 
++ (3,0) node {\huge ...} +(1,0) node[anchor=west] { $M =  (-\infty,\beta_0)\cup [ \beta_1,\beta_2) \cup [ \beta_3,\beta_4)$};

\draw[line width=1pt] (4,1) -- ++ (0,1.5);

\foreach \x  in {-3,...,11} 	\draw (\x+4.5,2.5)  node {$\x$};
\path (6.5,0.5) node {$\beta_0$} ++ (1,0) node {$\beta_1$}
++ (2,0) node {$\beta_2$}++ (2,0) node {$\beta_3$}++ (3,0) node {$\beta_4$}
;
\end{tikzpicture}

\caption{Block coordinates $\bbeta=\{2,3,5,7,10\}$
  of a genus $2$ Maya diagram. }\label{fig:genusM}
\end{figure}

\subsection{Modular decomposition and colouring.}

We have just seen how to characterize cyclic Maya diagrams $M\in\cM_g$
with shift $k=1$: one needs $2g+1$ flips given by the block
coordinates of $M$. The following dual notions of \textit{interlacing}
and \textit{modular decomposition} allow us to leverage this result to
arbitrary shifts $k$.  Note that, due to the reversal symmetry
\eqref{eq:reversal2}, we can restrict the analysis to positive shifts
$k>0$ without loss of generality.

The concept of block coordinates from Definition~\ref{def:genus} can
be extended naturally to Maya diagrams expressed as the interlacing of
$k$ Maya diagrams $M^{(0)}, \ldots, M^{(k-1)}$.  To describe $M$ using
the block coordinates of each $M^{(j)},\; j=0,\ldots, k-1$, we
introduce the notion of coloured multisets.  In the following section
we will extend this idea to the notion of coloured sequences, which
will furnish us with a combinatorial representation of rational
solutions compatible with the action of the extended affine Weyl group
$\EAWeyln$.

\begin{definition}\label{def:interlacing}
  Fix a $k\in \N$ and let $M^{(0)}, M^{(1)},\ldots M^{(k-1)}\subset \Z$ be sets
  of integers.  We define the \textit{interlacing of $k$ sets} to be the set
  \begin{equation}\label{eq:interlacing}
    \Theta_k\left(M^{(0)},
      M^{(1)},\ldots M^{(k-1)}\right) 
    = \bigcup_{i=0}^{k-1} (k M^{(i)} +i),
 \end{equation}
 where
 \[ kM +j = \{ km + j \colon m\in M \},\quad M\subset \Z.\]
Conversely, given a set of integers $M\subset \Z$ and a positive integer $k\in \N$, we can define the \textit{$k$-modular decomposition} of $M$ as the $k$-tuple of sets $\left(M^{(0)}, M^{(1)},\ldots M^{(k-1)}\right)$, where
 \[ M^{(i)} = \{ m\in \Z \colon km+i \in M\},\quad i=0,1,\ldots, k-1.\]
\end{definition}
\noindent
These operations are clearly the inverse of each other, in the sense
that $\left(M^{(0)}, M^{(1)},\ldots M^{(k-1)}\right)$ is the
$k$-modular decomposition of $M$ if and only if
$M=\Theta_k\lp M^{(0)}, M^{(1)},\ldots M^{(k-1)}\rp$.

Although interlacing and modular decomposition apply to general sets,
they have a well defined restriction to Maya diagrams. Indeed, if
$M=\Theta_k\left(M^{(0)}, M^{(1)},\ldots M^{(k-1)}\right)$ and $M$ is
a Maya diagram, then $M^{(0)}, M^{(1)},\ldots M^{(k-1)}$ are also Maya
diagrams. The converse is also true.

The notions of interlacing and modular decomposition also apply to the
setting of finite integer multisets.  Let $A \sqcup B$ denote the
disjoint union of multisets $A,B$; i.e. $\sqcup$ is the operation that
adds element multiplicities.  It is clear that if
$\bgamma^{(i)} \in \cZ^{p_i},\quad i=0,1,\ldots, k-1$, then 

\begin{align}
  \label{eq:bbeTheta}
  \bgamma^{[k]}=\Theta_k(\bgamma^{(0)},\ldots, \bgamma^{(k-1)})
  = \bigsqcup_{i=0}^{k-1} (k \bgamma^{(i)} +i),
\end{align}
is an integer multiset of cardinality
\[ p=p_0+\cdots + p_{k-1},\] and that
$(\bgamma^{(0)},\ldots, \bgamma^{(k-1)})$ serve as the $k$-modular
decomposition of $\bgamma^{[k]}$.

It should be noted that modular decompositions of Maya diagrams have
been considered previously by Noumi (Proposition 7.12 in
\cite{noumi2004painleve}), although in a slightly different context.
Noumi describes the action of of B\"acklund transformations on
  single Maya diagrams, while in our context it is convenient to
  describe the corresponding action on the full Maya cycle. We shall
  address this matter further in Section~\ref{sec:sym}. After
introducing the notions of genus and interlacing, we recall without
proof the main result to characterize cyclic Maya diagrams.

\begin{prop}[Theorem 4.8 in \cite{clarkson2020cyclic}]
  \label{thm:Mp}
  Consider an arbitrary Maya diagram $M$, let $M=\Theta_k\left(M^{(0)}, M^{(1)},\ldots M^{(k-1)}\right)$ be its $k$-modular decomposition, and $g_i$ the genus of $M^{(i)}$ for $i=0,1,\ldots, k-1$.  Then, $M$ is $(p,k)$-cyclic where
  \begin{equation}
    \label{eq:pgi}
    p = p_0+p_1+\cdots + p_{k-1},\qquad p_i = 2g_i + 1.
  \end{equation}
\end{prop}

The proof of this Proposition essentially states that a shift of $M$
by $k$ can only be done if each of the $M^{(i)}$ is shifted by one,
for which precisely $p$ flip operations at locations
\eqref{eq:bbeTheta} are needed \cite{clarkson2020cyclic}.
Proposition~\ref{thm:Mp} establishes a link between the shift $k$, the
period $p$ and the genera of the Maya diagrams that form the
$k$-modular decomposition of $M$. Applying this Proposition for a
fixed period $p=2n+1$, one can enumerate all possible cyclic Maya
diagrams with that period, so it is a key element towards the full
classification of rational solutions to the dressing chain. From \eqref{eq:reversal2} and \eqref{eq:pgi} we see the only possible values of the shift $\Delta$ for an odd-cyclic dressing chain.
\begin{cor}\label{cor:k}
  For a fixed period $p=2n+1\in \N$, there exist $(2n+1)$-cyclic Maya diagrams with
  shifts $k=\pm (2n+1),\pm(2n-1),\dots,1$, and no other shifts are  possible.
\end{cor}

\begin{remark}
  The highest shift $k=p$ corresponds to the interlacing of $p$ trivial (genus
  0) Maya diagrams. This class of solutions has been described already by Tsuda \cite{tsuda2005universal}, where the interlacing of $p$ genus-0 Maya diagrams correspond to $p$-reduced partitions. For $p=3$, these solutions are known as Okamoto polynomials, so in general the highest shift $k=2n+1$ dressing chains generalize the Okamoto class.

\end{remark}

We now introduce the notion of colouring, a useful visual
representation of modular decomposition.  Colouring also plays an
essential role in the formulation of the classification results that
follow.

\begin{definition}
  \label{def:setcolour}
  A $k$-coloured multiset is the assignment of one of $k$ colours to
  the elements of a given integer multiset. Formally, we represent a $k$-coloured multiset by $\bgamma=\{( \gamma_i,C_i)\}_{i=1}^p$ where $\gamma_i\in\Z$ and
  $C_i\in\ZkZ=\{0,1,\ldots, k-1\}$ is the ``colour'' of the $i$-th element. A $k$-coloured multiset defines the following multiset  decomposition
  \begin{equation}\label{eq:gammadecomp}
    \bgamma =  \bgamma^{(0)}\sqcup \cdots \sqcup \bgamma^{(k-1)},
  \end{equation}
  where $\bgamma^{(j)},\; j\in \ZkZ$ is the
  sub-multiset of elements having colour $j$.  We  use $\cZp_k$ to
  denote the set of all $k$-coloured multisets of cardinality $p$. Let
  $p_j,\; j\in \ZkZ$ denote the cardinality of $\bgamma^{(j)}$.  We
  will call the sequence $\bp=(p_0,\ldots, p_{k-1})$ the signature of
  $\bgamma$.  Observe that, by definition, $\bp$ serves as a
  composition of $p$, namely $p=p_0+\cdots+ p_{k-1}$.
\end{definition}
\noindent
We may now express the interlacing operator $\Theta_k$ defined in
\eqref{eq:bbeTheta} as the bijection $\Theta_k\colon \cZ^p_k\to \cZ^p$
with action given by \eqref{eq:bbeTheta}

\begin{definition}
  \label{def:flipset}
  Fix a $k\in \N$ and let $M$ be a Maya diagram. We refer to
  \begin{equation}
    \label{eq:bmudef}
    \bgamma^{[k]}=(M+k)\ominus M
  \end{equation}
  as the \textit{$k\supth$ order flip set} of $M$. We will call
  $\bgamma = \Theta_k^{-1}(\bgamma^{[k]})$ the $k\supth$ order block
  coordinates of $M$, and refer to the corresponding
  $\bp=(p_0,\ldots, p_{k-1})$ as the $k\supth$ order signature of $M$.
\end{definition}
\noindent
Observe that \eqref{eq:bmudef} entails
\begin{equation}
  \label{eq:bmuMk}
  \phi_{\bgamma^{[k]}}(M) = M+k 
\end{equation}

Thus, the \textit{$k\supth$ order flip set} $\bgamma^{[k]}$  is the minimum set of flips that turns $M$ into $M+k$. This set coincides with the block coordinates $\bbeta$ of $M$  when $k=1$, but otherwise the two sets are different.

\begin{definition}
  We say that that $\bgamma\in \cZ^p_k$ is an oddly coloured multiset
  if the entries of the corresponding signature $\bp$ are odd, that is
  if each colour occurs an odd number of times\footnote{If $p$ is odd,
    then the number of colours $k$ in an odd colouring must also be
    odd.}.  We  say that $\bgamma\in \cZp_k$ is a $k$-coloured set
  if the corresponding $\bgamma^{[k]}=\Theta_k(\gamma,C)$ is a set, or equivalently if each of the $\gamma^{(j)}$ in the decomposition \eqref{eq:gammadecomp} do not contain repeated elements.

\end{definition}

Note that for a given oddly coloured multiset $\bgamma\in\cZ^p_k$, $\bgamma^{[k]}\in\Z^p$ defined by \eqref{eq:bbeTheta} is in general a multi-set such that \eqref{eq:bmuMk} holds. However, the $k\supth$ order flip set defined by \eqref{eq:bmudef} is always a set as it contains no repeated integers.

\begin{prop}[Proposition 4.13 in \cite{clarkson2020cyclic}]
  Fix a $k\in \N$.  For every Maya diagram $M$, the corresponding
  $k\supth$ order block coordinates are an oddly $k$-coloured
  set. Conversely, for an oddly coloured multiset
  $\bgamma\in \cZ^p_k$, define
  \begin{equation}
    \label{eq:Xikdef}
     \Xi_k(\bgamma) =  \Theta_k(\Xi(\bgamma^{(0)}), \ldots,
    \Xi(\bgamma^{(k-1)})). 
  \end{equation}
  Then, $M=\Xi_k(\bgamma)$ is a $(p,k)$-cyclic Maya diagram.  
\end{prop}
\begin{proof}
  Let $M$ be a Maya diagram and $\bgamma^{[k]}=(M+k)\ominus M$ its $k\supth$
  order flip set.  Set $\bgamma=\Theta_k^{-1}(\bgamma^{[k]})$ and observe
  that
  \begin{align*}
    \phi_{\bgamma^{[k]}}(M)                    &= M+k\\
    &= \Theta_k(\phi_{\bgamma^{(0)}}M^{(0)} ,\ldots,
                     \phi_{\bgamma^{(k-1)}}M^{(k-1)})\\
    &= \Theta_k(M^{(0)} + 1,\ldots, M^{(k-1)}+1) 
  \end{align*}
  It follows that
  \[ M^{(j)}= \Xi(\bgamma^{(j)}),\quad j\in \Z/k\Z,\] and hence that
  each $\bgamma^{(j)},\; j\in \ZkZ$ has odd cardinality.  

  We turn to the proof of the converse. Suppose that
  $\bgamma\in \cZ^p_k$ is an oddly coloured multiset, and let
  $M=\Xi_k(\bgamma)$.  Set $\bgamma^{[k]}= \Theta_k(\bgamma)$ and observe
  that, by construction, $\phi_{\bgamma^{[k]}}(M) = M+k$. This proves the
  second assertion.
\end{proof}

\begin{example}\label{ex:1}
  Figure~\ref{fig:interlacing} provides a visual interpretation of the
  modular decomposition of a Maya diagram $M$ into Maya
  diagrams $M^{(0)}, M^{(1)}, M^{(2)}$ of genus $1,2,0$, respectively.
  Each of these Maya diagrams is dilated by a factor of $3$, shifted by one unit
  with respect to the previous one and superimposed.

The block coordinates of each of the three diagrams are given by:  
 \[
\begin{aligned}
&\gamma^{(0)}=\{0,1,4\},&\qquad &M^{(0)}=\Xi(\gamma^{(0)})=(-\infty,0)\cup[1,4)\\
&\gamma^{(1)}=\{-1,1,3,5,6\},&\qquad &M^{(1)}=\Xi(\gamma^{(1)})=(-\infty,-1)\cup[1,3)\cup[5,6)\\
&\gamma^{(2)}=\{5\},&\qquad &M^{(2)}=\Xi(\gamma^{(2)})=(-\infty,5)
\end{aligned}
\]
The set of colours is $\Z/3\Z=\{0,\rd{1},\bl{2}  \}$. The interlacing of these three Maya diagrams is described by the $3^{\textrm{rd}}$ order block coordinates
\[ \bgamma= \gamma^{(0)}\sqcup \gamma^{(1)}\sqcup \gamma^{(0)} = \{ 0,1,4,\rd{-1,1,3,5,6}, \bl{5}   \},  \]  
which form a $3$-coloured set of cardinality $p=p_0+p_1+p_2=3+5+1=9$. The signature is therefore $\bp=(3,5,1)$. The $3^{\textrm{rd}}$ order flip set is given by 
\[ \bgamma^{[3]}=\Theta_3(\{0,1,4,\rd{-1,1,3,5,6}, \bl{5}\})=
  \{0,3,12,-2,4,10,16,19,17 \}\in \cZ^{9}\] It is straightforward to
verify in this example that \eqref{eq:bmuMk} holds. Note that the
interlaced diagram $M$ has genus $5$ and its block coordinates
$\bbeta$ are given by
\[  \bbeta =\{  -2,-1,0,2,10,11,12,14,15,16,17  \}\in\hat\cZ^{11} \]
In general, there are no simple expressions to derive $\bbeta$ from $\bgamma^{(i)}$ or to connect the genera $g_i$ of the coloured Maya diagrams $M^{(i)}$ with the genus $g$ of the resulting interlaced Maya diagram $M$. However, the block coordinates $\bbeta$ of interlaced Maya diagrams $M$ do not play any significant role in the construction of Maya cycles and rational solutions.
\end{example}

\begin{figure}[h]
\begin{tikzpicture}[scale=0.6]

\draw  (1,3) grid +(11 ,1);

\path [fill,color=black] (0.5,3.5) node {\huge ...}  ++(1,0) circle
(5pt) ++(1,0) circle (5pt) ++(1,0) circle (5pt) ++(1,0) circle (5pt)
++(2,0) circle (5pt) ++(1,0) circle (5pt) ++(1,0) circle (5pt) ++
(4,0) node {\huge ...} +(1,0) node[anchor=west,color=black] {
  $M_0 = \Xi(\{0,1,4\}),\;\qquad\quad\,\,\,\, g_0 = 1$};

\draw[line width=1pt] (5,3) -- ++ (0,2);

\foreach \x in {-4,...,6} 	\draw (\x+5.5,4.5)  node {$\x$};

\draw  (1,1) grid +(11 ,1);

\path [fill,color=red] (0.5,1.5) node {\huge ...} 
++(1,0) circle (5pt) ++(1,0) circle (5pt)  ++(1,0) circle (5pt) 
++(3,0) circle (5pt) ++(1,0) circle (5pt) 
++ (3,0) circle (5pt)  ++(2,0) node {\huge ...} +(1,0) node[anchor=west,color=black] { $M_1 = \Xi(\{-1,1,3,5,6\}),\;\quad g_1 = 2$}; 

\draw[line width=1pt] (5,1) -- ++ (0,2);

\draw  (1,-1) grid +(11 ,1);

\path [fill,color=blue] (0.5,-0.5) node {\huge ...}  ++(1,0) circle
(5pt) ++(1,0) circle (5pt) ++(1,0) circle (5pt) ++(1,0) circle (5pt)
++(1,0) circle (5pt) ++(1,0) circle (5pt) ++(1,0) circle (5pt) ++(1,0)
circle (5pt) ++(1,0) circle (5pt) ++(3,0) node {\huge ...} +(1,0)
node[anchor=west,color=black] { $M_2 = \Xi(\{5\}),\qquad\qquad\qquad g_2 = 0$};

\draw[line width=1pt] (5,-1) -- ++ (0,2);

\draw  (0,-4) grid +(23 ,1);
\foreach \x in {-5,...,17} 	\draw (\x+5.5,2.5-5)  node {$\x$};
\draw[line width=1pt] (5,-4) -- ++ (0,2);

\path [fill,color=black] (2.5,-3.5)   
circle (5pt) ++(6,0) circle (5pt)  
++(3,0) circle (5pt) ++(3,0) circle (5pt) ;

\path [fill,color=red] 
(0.5,-3.5) circle (5pt) ++(9,0) circle (5pt)  
++(3,0) circle (5pt) ++(9,0) circle (5pt) ;

\path [fill,color=blue] 
(1.5,-3.5) circle (5pt) ++(3,0) circle (5pt)  
++(3,0) circle (5pt) ++(3,0) circle (5pt)
++(3,0) circle (5pt) ++(3,0) circle (5pt)  ++(3,0) circle (5pt) ;

\draw (3.5,-5) node[right]
{$M= \Theta_3(M_0,M_1,M_2) = \Xi_3(\{0,1,4,\rd{-1,1,3,5,6},
  \bl{5}\})$}; \draw (3.5,-6) node[right]
{};
\end{tikzpicture}

\caption{Interlacing of three Maya diagrams with genus $1,2$ and $0$}\label{fig:interlacing}
\end{figure}
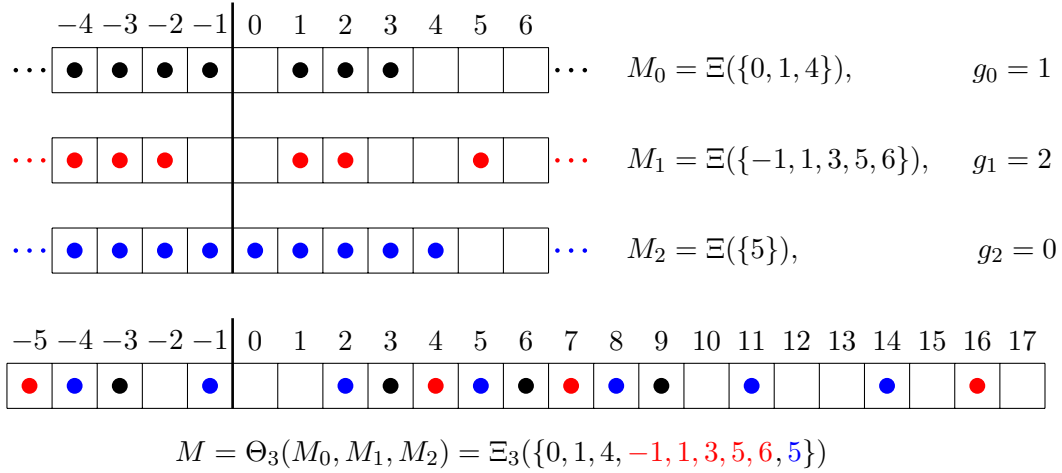

\section{Classification of rational solutions to $A_{2n}$-Painlev\'e}\label{sec:classification}

Definition \ref{def:mayacycle} above introduced the concept of
\textit{Maya cycles} : sequences of Maya diagrams connected by flip
operations that close into a cycle.  In order to build all rational
solutions of a $(2n+1)$ cyclic dressing chain, all we need to specify
is how to build all $(2n+1,k)$ Maya cycles for
$k=\pm 1,\dots,\pm 2n+1$.

Evidently, every $M_i$ in a $(p,k)$ Maya cycle is $(p,k)$-cyclic as
per Definition \ref{def:cyclicM}.  The following Proposition
elucidates the relationship between cyclic Maya diagrams and Maya
cycles.
\begin{prop}
  Let $M_0$ be a $(p,k)$-cyclic Maya diagram,
  $\phi_{\bgamma^{[k]}},\; \bgamma^{[k]}\in \cZp$ a multi-flip such that
  $\phi_{\bgamma^{[k]}}(M_0) = M_0+k$, and $\bmu\in \Z^p$ an arbitrary
  enumeration of $\bgamma^{[k]}$. Then
  \begin{equation}
    M_{i+1}=\phi_{\mu_i}(M_{i}),\qquad i=0,\dots,p-1,
  \end{equation} defines a $(p,k)$ Maya cycle.
\end{prop}

\subsection{Coloured sequences.}
\label{sect:colseq}
In the same way that  $(p,k)$ cyclic Maya diagrams can be indexed by $k$-coloured sets,  $(p,k)$ Maya cycles will be indexed by $k$-coloured sequences, a concept that we introduce next.
\begin{definition}
  A $k$-colouring of a sequence is the assignment of one of $k$
  colours to each component of that sequence.  Setting
  $\Zp_k= \Z^p\times (\ZkZ)^p$, we formally represent a
  \textit{$k$-coloured integer sequence }of length $p$ as a pair
  \[ (\bnu,C)=\big((\nu_0,\ldots, \nu_{p-1}),(C_0,\dots,C_{p-1})\big)\in \Zp_k\] 
where $C_i$ is the colour of $\nu_i$, for  $i=0,\ldots, p-1$. As before, we say that $(\bnu,C)$ is
  oddly coloured if each colour occurs an odd number of times.  

\end{definition}
Given a coloured sequence $(\bnu,C)\in \Zp_k$, let $[\bnu,C]\in \cZp_k$ be
  the corresponding $k$-coloured multiset whose elements are the
  components of the sequence in question.  Formally,
  \begin{equation}\label{eq:fromnu2gamma}
   [\bnu,C] = (\bgamma^{(0)},\ldots, \bgamma^{(k-1)}),
  \end{equation}
  where
  \begin{equation}\label{eq:fromnu2gamma2}
   \bgamma^{(j)}:= \{ \nu_i \colon C_i = j \},\quad j\in \ZkZ.
   \end{equation}
  
\begin{definition}\label{def:shift}
Define the shift operator $\pi:\Zp_k\to \Zp_k$ with action
\begin{equation}
  \label{eq:piaction}
  \pi(\bnu,C) = (L(\bnu)+\be_{p-1},L(C)),\quad \bnu\in \Z^p,\; C\in (\ZkZ)^p,
\end{equation}
where $L$ is the circular permutation
\begin{equation}
  \label{eq:Ldef}
   L(C) = (C_1,\ldots, C_{p-1},C_0), 
\end{equation}
and $\be_{i}\in \Z^p,\;i=0,\ldots, p-1$ is the $i\supth$ unit vector. Thus,
\begin{eqnarray*}   \pi(\bnu,C) &=&  \pi\big((\nu_0,\ldots, \nu_{p-1}),(C_0,\dots,C_{p-1})\big) \\
&=& \big( (\nu_1,\nu_2,\ldots,
  \nu_{p-1},\nu_0+1), (C_1,\dots,C_{p-1},C_0)\big).
  \end{eqnarray*}
\end{definition}

The next Proposition describes the correspondence between coloured
sequences and Maya cycles.  It makes use of the following auxilliary
notation. Consider
\begin{equation}
  \label{eq:Xikdef2}
  \Xik(\bnu,C) := \Xik([\bnu,C]),\quad (\bnu,C)\in \Zp_k.
\end{equation}
as the generalization of  \eqref{eq:Xikdef} from coloured multisets to coloured sequences, and let 
\[ \pi^i = \overbrace{\pi\circ \cdots \circ \pi}^{i\; \text{times}}\]
denote the iterated action of $\pi$ as defined by $\eqref{eq:piaction}$.

\begin{prop}
  \label{prop:nuCM}
  Let $(\bnu,C)\in \Zp_k$ be an oddly coloured integer sequence, and
  set
  \begin{align}
    \label{eq:mufromnu}
    \mu_i &= k \nu_i + C_i,\quad i=0,\ldots, p-1\\
    \label{eq:M0nu}
    M_0 &= \Xik(\bnu,C),\\
    \label{eq:Mipi}
    M_{i+1} &= \phi_{\mu_i}(M_i),\quad i=0,\ldots, p-1.
  \end{align}
   Then, $\bM=(M_0,\ldots, M_{p})$ is a $(p,k)$ Maya cycle with flip
   sequence $\bmu$ that satisfies
   \begin{align}
     M_i=\Xi_k(\pi^i(\bnu,C)),\quad i = 0,\ldots, p,
   \end{align}
   The above mapping $(\bnu,C)\mapsto \bM,\; (\bnu,C)\in \Zp_k$
   constitutes a bijection between the set of oddly $k$-coloured
   sequences of length $p$ and the set of $(p,k)$ Maya cycles.  The
   inverse mapping $\bM\to (\bnu,C)$ is given by taking the $k$-modular
   decomposition of the  flip sequence corresponding to $\bM$.
\end{prop}
\begin{proof}
The proof follows by a straightforward application of the relevant definitions.
\end{proof}

\begin{definition}
  In parallel to the terminology introduced above, we will refer to
  the coloured integer sequence $(\bnu,C)$ as the block coordinates of
  the Maya cycle generated by \eqref{eq:Mipi}.
\end{definition}

We next describe the effect of translations on a Maya cycle at the level of the coloured sequences, in order to define an equivalence class under translations.  
Let $T\colon \Zp_k\to \Zp_k$ be the invertible mapping
defined by
\[ T: (\bnu,C)\mapsto (\hbnu,L^{-1}(C)),\quad \bnu\in \Z^p,\; C\in
  (\ZkZ)^p,\]
where
\begin{equation}
  \label{eq:M+1bC}
  \hnu_i =
  \begin{cases}
    \nu_{i}+1 & \text{ if } C_i = k-1\\
    \nu_i & \text{ otherwise,}
  \end{cases}
\end{equation}
and where $L$ is the circular permutation \eqref{eq:Ldef}.
\begin{prop}\label{prop:M+1}
  Let $(\bnu,C)\in \Zp_k$ be the block coordinates of a $(p,k)$ Maya
  cycle $\bM$. Then $T(\bnu,C)$ are the block coordinates of the Maya
  cycle $\bM+1=(M_0+1,M_1+1,\ldots, M_p+1)$.
\end{prop}

\begin{prop}\label{prop:MM+1}
Maya cycles $\bM=(M_0,\dots,M_{2n+1})$ and $\bM+1=(M_0+1,\dots,M_{2n+1}+1)$ generate the same rational solution $(w_0,\dots,w_{2n}|a_0,\dots,a_{2n})$ of a $2n+1$-cyclic dressing chain.
\end{prop}
\begin{proof}
The proof comes from a straightforward application of the construction formulas \eqref{eq:wsigMi}-\eqref{eq:aimui}. We recall that for any Maya diagram $M$, the pseudo-Wronskians $H_M$ and $H_{M+1}$ only differ by a multiplicative constant, as seen in Proposition~\ref{prop:Mshift}. Since only log-derivatives of  Hermite pseudo-Wronskians enter in the rational solution  \eqref{eq:wsigMi} and the parameters \eqref{eq:aimui} only involve differences of the components of the flip sequence, an overall translation of the Maya cycle has no effect in the rational solution.
\end{proof}

The last two Propositions imply that there is an equivalence class of Maya cycles related by translations that generate the same rational solution of a dressing chain. The correspondence between  coloured sequences and rational solutions is thus many to one. A one-to-one correspondence can be achieved by  fixing a canonical representative in each equivalence class. 
\begin{definition}
A $(p,k)$ Maya cycle with $k\in\N$ is in standard form if and only if its first diagram $M_0$ is in standard form. A coloured sequence $(\bnu,C)$ is in standard form if the Maya cycle it defines by \eqref{eq:mufromnu}-\eqref{eq:Mipi} in standard form too.
\end{definition}

It is obvious that in each equivalence class $\bM/\Z$ of Maya cycles related by translations, only one of them is in standard form.

\begin{example}
The $(5,3)$ Maya cycle $\bM$ defined by the coloured sequence $(\bl{4},\rd{3},\bl{1},\bl{2},0)$ is in standard form. The action of a unit translation gives a Maya cycle $\bM+1$ which is not in standard form. Both of them are shown in Figure~\ref{fig:M+1}, where the action of $T$ on the block coordinates described by Proposition~\ref{prop:M+1} can be verified.

\begin{figure}
  \centering
\begin{tikzpicture}[scale=0.5]
  \path [fill,red] (0.5,6.5)  
  ++(2,0) circle (5pt)
  ++(3,0) circle (5pt)
  ++(3,0) circle (5pt) ;

  \path [fill,blue] (0.5,6.5)  
  ++(3,0) circle (5pt)
  ++(6,0) circle (5pt) 
  ++(3,0) circle (5pt) ;

  \path (17.5,6.5)  node[anchor=west] {$(\bl{4},\rd{3},\bl{1},\bl{2},0)$};

  \path [fill,red] (0.5,5.5)  
  ++(2,0) circle (5pt)
  ++(3,0) circle (5pt)
  ++(3,0) circle (5pt) ;

  \path [fill,blue] (0.5,5.5)  
  ++(3,0) circle (5pt)
  ++(6,0) circle (5pt) 
  ++(3,0) circle (5pt) 
  ++(3,0) circle (5pt) ;
  \path (17.5,5.5)  node[anchor=west] {$(\rd{3},\bl{1},\bl{2},0,\bl{5})$};

  \path [fill,red] (0.5,4.5)  
  ++(2,0) circle (5pt)
  ++(3,0) circle (5pt)
  ++(3,0) circle (5pt)
  ++(3,0) circle (5pt) ;

  \path [fill,blue] (0.5,4.5)  
  ++(3,0) circle (5pt)
  ++(6,0) circle (5pt) 
  ++(3,0) circle (5pt) 
  ++(3,0) circle (5pt) ;
  \path (17.5,4.5)  node[anchor=west] {$(\bl{1},\bl{2},0,\bl{5},\rd{4})$};

  \path [fill,red] (0.5,3.5)  
  ++(2,0) circle (5pt)
  ++(3,0) circle (5pt)
  ++(3,0) circle (5pt)
  ++(3,0) circle (5pt) ;

  \path [fill,blue] (0.5,3.5)  
  ++(3,0) circle (5pt)
  ++(3,0) circle (5pt)
  ++(3,0) circle (5pt) 
  ++(3,0) circle (5pt) 
  ++(3,0) circle (5pt) ;
  \path (17.5,3.5)  node[anchor=west] {$(\bl{2},0,\bl{5},\rd{4},\bl{2})$};

  \path [fill,red] (0.5,2.5)  
  ++(2,0) circle (5pt)
  ++(3,0) circle (5pt)
  ++(3,0) circle (5pt)
  ++(3,0) circle (5pt) ;
  \path [fill,blue] (0.5,2.5)  
  ++(3,0) circle (5pt)
  ++(3,0) circle (5pt)
  ++(6,0) circle (5pt) 
  ++(3,0) circle (5pt) ;
  \path (17.5,2.5)  node[anchor=west] {$(0,\bl{5},\rd{4},\bl{2},\bl{3})$};

  \path [fill,black] (0.5,1.5)  
  ++(1,0) circle (5pt);  
  \path [fill,red] (0.5,1.5)  
  ++(2,0) circle (5pt)
  ++(3,0) circle (5pt)
  ++(3,0) circle (5pt)
  ++(3,0) circle (5pt) ;
  \path [fill,blue] (0.5,1.5)  
  ++(3,0) circle (5pt)
  ++(3,0) circle (5pt)
  ++(6,0) circle (5pt) 
  ++(3,0) circle (5pt) ;
  \path (17.5,1.5)  node[anchor=west] {$(\bl{5},\rd{4},\bl{2},\bl{3},1)$};

  \path [fill,blue] (0.5,1.5) ++(0,0) circle (5pt)
  ++(0,1) circle (5pt)
  ++(0,1) circle (5pt)
  ++(0,1) circle (5pt)
  ++(0,1) circle (5pt)
  ++(0,1) circle (5pt);

  \path [fill,red] (-.5,1.5) ++(0,0) circle (5pt)
  ++(0,1) circle (5pt)
  ++(0,1) circle (5pt)
  ++(0,1) circle (5pt)
  ++(0,1) circle (5pt)
  ++(0,1) circle (5pt);

  \path [fill,black] (-1.5,1.5) ++(0,0) circle (5pt)
  ++(0,1) circle (5pt)
  ++(0,1) circle (5pt)
  ++(0,1) circle (5pt)
  ++(0,1) circle (5pt)
  ++(0,1) circle (5pt);

  \draw  (-2,1) grid +(19 ,6);
  \draw[line width=2pt] (1,1) -- ++ (0,6);
  \draw[line width=2pt,dashed] (16,1) -- ++ (0,6);

  \foreach \x in {0,...,15} \draw (\x+1.5,0.5)  node {$\x$};
\end{tikzpicture}

\begin{tikzpicture}[scale=0.5]
  \path [fill,blue] (0.5,6.5)  
  ++(2,0) circle (5pt)
  ++(3,0) circle (5pt)
  ++(3,0) circle (5pt) ;

  \path [fill,black] (0.5,6.5)  
  ++(3,0) circle (5pt)
  ++(6,0) circle (5pt) 
  ++(3,0) circle (5pt) ;

  \path (17.5,6.5)  node[anchor=west] {$(5,\bl{3},2,3,\rd{0})$};

  \path [fill,blue] (0.5,5.5)  
  ++(2,0) circle (5pt)
  ++(3,0) circle (5pt)
  ++(3,0) circle (5pt) ;

  \path [fill,black] (0.5,5.5)  
  ++(3,0) circle (5pt)
  ++(6,0) circle (5pt) 
  ++(3,0) circle (5pt) 
  ++(3,0) circle (5pt) ;
  \path (17.5,5.5)  node[anchor=west] {$(\bl{3},2,3,\rd{0},6)$};

  \path [fill,blue] (0.5,4.5)  
  ++(2,0) circle (5pt)
  ++(3,0) circle (5pt)
  ++(3,0) circle (5pt)
  ++(3,0) circle (5pt) ;

  \path [fill,black] (0.5,4.5)  
  ++(3,0) circle (5pt)
  ++(6,0) circle (5pt) 
  ++(3,0) circle (5pt) 
  ++(3,0) circle (5pt) ;
  \path (17.5,4.5)  node[anchor=west] {$(2,3,\rd{0},6,\bl{4})$};

  \path [fill,blue] (0.5,3.5)  
  ++(2,0) circle (5pt)
  ++(3,0) circle (5pt)
  ++(3,0) circle (5pt)
  ++(3,0) circle (5pt) ;

  \path [fill,black] (0.5,3.5)  
  ++(3,0) circle (5pt)
  ++(3,0) circle (5pt)
  ++(3,0) circle (5pt) 
  ++(3,0) circle (5pt) 
  ++(3,0) circle (5pt) ;
  \path (17.5,3.5)  node[anchor=west] {$(3,\rd{0},6,\bl{4},3)$};

  \path [fill,blue] (0.5,2.5)  
  ++(2,0) circle (5pt)
  ++(3,0) circle (5pt)
  ++(3,0) circle (5pt)
  ++(3,0) circle (5pt) ;
  \path [fill,black] (0.5,2.5)  
  ++(3,0) circle (5pt)
  ++(3,0) circle (5pt)
  ++(6,0) circle (5pt) 
  ++(3,0) circle (5pt) ;
  \path (17.5,2.5)  node[anchor=west] {$(\rd{0},6,\bl{4},3,4)$};

  \path [fill,red] (0.5,1.5)  
  ++(1,0) circle (5pt);  
  \path [fill,blue] (0.5,1.5)  
  ++(2,0) circle (5pt)
  ++(3,0) circle (5pt)
  ++(3,0) circle (5pt)
  ++(3,0) circle (5pt) ;
  \path [fill,black] (0.5,1.5)  
  ++(3,0) circle (5pt)
  ++(3,0) circle (5pt)
  ++(6,0) circle (5pt) 
  ++(3,0) circle (5pt) ;
  \path (17.5,1.5)  node[anchor=west]  {$(6,\bl{4},3,4,\rd{1})$};

  \path [fill,black] (0.5,1.5) ++(0,0) circle (5pt)
  ++(0,1) circle (5pt)
  ++(0,1) circle (5pt)
  ++(0,1) circle (5pt)
  ++(0,1) circle (5pt)
  ++(0,1) circle (5pt);

  \path [fill,blue] (-.5,1.5) ++(0,0) circle (5pt)
  ++(0,1) circle (5pt)
  ++(0,1) circle (5pt)
  ++(0,1) circle (5pt)
  ++(0,1) circle (5pt)
  ++(0,1) circle (5pt);

  \path [fill,red] (-1.5,1.5) ++(0,0) circle (5pt)
  ++(0,1) circle (5pt)
  ++(0,1) circle (5pt)
  ++(0,1) circle (5pt)
  ++(0,1) circle (5pt)
  ++(0,1) circle (5pt);

  \draw  (-2,1) grid +(19 ,6);
  \draw[line width=2pt] (0,1) -- ++ (0,6);
  \draw[line width=2pt,dashed] (15,1) -- ++ (0,6);

  \foreach \x in {0,...,15} \draw (\x+0.5,0.5)  node {$\x$};
\end{tikzpicture}

\caption{The effect of a unit translation $T$ on a $(5,3)$ Maya
  cycle. }
\label{fig:M+1}
\end{figure}
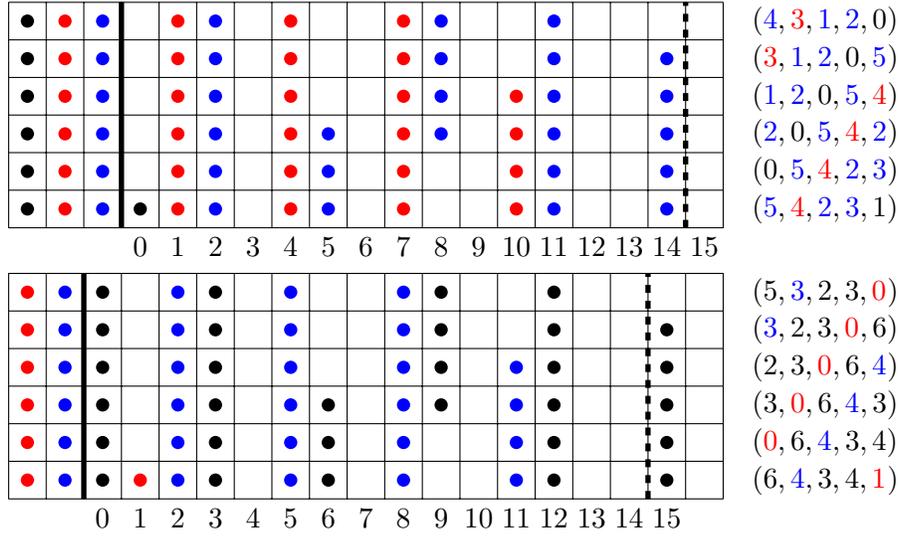
\end{example}

We can finally state the main theorem that expresses a good indexing scheme for rational solutions to an odd cyclic dressing chain.

\begin{thm}\label{thm:main}
 The set of rational solutions to a  $(2n+1)$-cyclic dressing chain with
  shift $\Delta=2k$ bijectively corresponds to the set of oddly
  $k$-coloured sequences $(\bnu,C)\in \Z^{2n+1}_k$ in standard form.
\end{thm}

\begin{proof}
Proposition~\ref{prop:nuCM} establishes a bijection between $(2n+1,k)$ Maya cycles and oddly coloured sequences $(\bnu,C)\in\Z^{2n+1}_k$. Theorem~\ref{thm:characterization} establishes a bijection between rational solutions of a $(2n+1)$-cyclic dressing chain with shift $\Delta=2k$ and $(2n+1,k)$ Maya cycles, up to a translation of the cycle. If the Maya cycle is required to be in standard form, the correspondence between rational solutions and oddly
  $k$-coloured sequences is one-to-one.
\end{proof}

\subsection{Enumeration and construction of  explicit examples}\label{sec:examples}

We shall describe now how to enumerate and construct explicitly all
rational solutions to the $(2n+1)$-cyclic dressing chain system
\eqref{eq:wfchain}, and therefore, by the equivalence described in
Proposition~\ref{prop:wtof}, also all rational solutions of the
$A_{2n}$-\p\ system \eqref{eq:A2nsystem}.

For a given cyclicity of the chain $2n+1$, by Corollary~\ref{cor:k} we
see that the only possible shifts are $k=\pm 1,\dots,\pm (2n+1)$. The
reversal symmetry \eqref{eq:reversal2} allows to invert the sign of
the shift, so we can focus without loss of generality on solutions
with positive shifts $k=1,\dots,2n+1$.

Next, we fix a given $k$ in that range, and ask ourselves how many different $k$-signatures must be considered. This is the number of different compositions of length $k$ of an odd number $2n+1$ with odd parts, which is precisely
\begin{equation}
a(2n+1,k)=\binom{n+\frac{k-1}{2}}{k-1},\qquad k=1,\dots,2n+1.
\end{equation}
The total number of possible signatures for a given period $2n+1$ is
\begin{equation}
\sum_{k=1}^{2n+1} a(2n+1,k)=F_{2n+1}
\end{equation}
where $F_j$ is the $j^{\textrm{th}}$ Fibonacci number.
As an example, an enumeration of all the possible signatures for $5$-periodic chains is:
\[
\begin{aligned}
&k=1: & a(5,1)=1, &\qquad  (5)\\
&k=3: & a(5,3)=3, &\qquad  (3,1,1), (1,3,1), (1,1,3)\\
&k=5: & a(5,5)=1, &\qquad  (1,1,1,1,1)\\
\end{aligned}
\]
for a total number of $F_3$=5.

\begin{example}

In order to construct a given rational solution, pick a shift and a
signature, say $k=3$ and $5=1+1+3$.  This means that the coloured
block coordinates are given by an integer 5-tuple grouped into 3
colours as per the above composition.  We assume, without loss of
generality, that $M_0$ is in standard form. Let us choose for
instance, $(\bnu,C) = (\bl{4},\rd{3},\bl{1},\bl{2},0)$, with the same colour code  $\Z/3\Z=\{0,\rd{1},\bl{2}  \}$ as in Example~\ref{ex:1}.

 Following
\eqref{eq:bbeTheta}, the $k\supth$ order flip set corresponds to
\[ \bgamma^{[3]}=\Theta_3(\{0,\bl{1,2},\rd{3},\bl{4}\})=\{0,5,8,
  10,14\} .\] As $(\bnu,C)$ does not contain repeated integers with
the same colour, this leads to a non-degenerate cycle (see Section)
and $\bgamma^{[3]}\in\hat\cZ^5$ is a set. In this non-degenerate setting,
every permutation of $\bgamma^{[3]}$ yields a different flip sequence
$\bmu$ and correspondingly a different $(5,3)$ Maya cycle.

\begin{figure}[ht]
  \centering

\begin{tikzpicture}[scale=0.5]
  \path [fill,red] (0.5,6.5)  
  ++(2,0) circle (5pt)
  ++(3,0) circle (5pt)
  ++(3,0) circle (5pt) ;

  \path [fill,blue] (0.5,6.5)  
  ++(3,0) circle (5pt)
  ++(6,0) circle (5pt) 
  ++(3,0) circle (5pt) ;

  \path (17.5,6.5)  node[anchor=west] {$(\bl{4},\rd{3},\bl{1},\bl{2},0)$};

  \path [fill,red] (0.5,5.5)  
  ++(2,0) circle (5pt)
  ++(3,0) circle (5pt)
  ++(3,0) circle (5pt) ;

  \path [fill,blue] (0.5,5.5)  
  ++(3,0) circle (5pt)
  ++(6,0) circle (5pt) 
  ++(3,0) circle (5pt) 
  ++(3,0) circle (5pt) ;
  \path (17.5,5.5)  node[anchor=west] {$(\rd{3},\bl{1},\bl{2},0,\bl{5})$};

  \path [fill,red] (0.5,4.5)  
  ++(2,0) circle (5pt)
  ++(3,0) circle (5pt)
  ++(3,0) circle (5pt)
  ++(3,0) circle (5pt) ;

  \path [fill,blue] (0.5,4.5)  
  ++(3,0) circle (5pt)
  ++(6,0) circle (5pt) 
  ++(3,0) circle (5pt) 
  ++(3,0) circle (5pt) ;
  \path (17.5,4.5)  node[anchor=west] {$(\bl{1},\bl{2},0,\bl{5},\rd{4})$};

  \path [fill,red] (0.5,3.5)  
  ++(2,0) circle (5pt)
  ++(3,0) circle (5pt)
  ++(3,0) circle (5pt)
  ++(3,0) circle (5pt) ;

  \path [fill,blue] (0.5,3.5)  
  ++(3,0) circle (5pt)
  ++(3,0) circle (5pt)
  ++(3,0) circle (5pt) 
  ++(3,0) circle (5pt) 
  ++(3,0) circle (5pt) ;
  \path (17.5,3.5)  node[anchor=west] {$(\bl{2},0,\bl{5},\rd{4},\bl{2})$};

  \path [fill,red] (0.5,2.5)  
  ++(2,0) circle (5pt)
  ++(3,0) circle (5pt)
  ++(3,0) circle (5pt)
  ++(3,0) circle (5pt) ;
  \path [fill,blue] (0.5,2.5)  
  ++(3,0) circle (5pt)
  ++(3,0) circle (5pt)
  ++(6,0) circle (5pt) 
  ++(3,0) circle (5pt) ;
  \path (17.5,2.5)  node[anchor=west] {$(0,\bl{5},\rd{4},\bl{2},\bl{3})$};

  \path [fill,black] (0.5,1.5)  
  ++(1,0) circle (5pt);  
  \path [fill,red] (0.5,1.5)  
  ++(2,0) circle (5pt)
  ++(3,0) circle (5pt)
  ++(3,0) circle (5pt)
  ++(3,0) circle (5pt) ;
  \path [fill,blue] (0.5,1.5)  
  ++(3,0) circle (5pt)
  ++(3,0) circle (5pt)
  ++(6,0) circle (5pt) 
  ++(3,0) circle (5pt) ;
  \path (17.5,1.5)  node[anchor=west] {$(\bl{5},\rd{4},\bl{2},\bl{3},1)$};

  \path [fill,blue] (0.5,1.5) ++(0,0) circle (5pt)
  ++(0,1) circle (5pt)
  ++(0,1) circle (5pt)
  ++(0,1) circle (5pt)
  ++(0,1) circle (5pt)
  ++(0,1) circle (5pt);

  \path [fill,red] (-.5,1.5) ++(0,0) circle (5pt)
  ++(0,1) circle (5pt)
  ++(0,1) circle (5pt)
  ++(0,1) circle (5pt)
  ++(0,1) circle (5pt)
  ++(0,1) circle (5pt);

  \path [fill,black] (-1.5,1.5) ++(0,0) circle (5pt)
  ++(0,1) circle (5pt)
  ++(0,1) circle (5pt)
  ++(0,1) circle (5pt)
  ++(0,1) circle (5pt)
  ++(0,1) circle (5pt);

  \draw  (-2,1) grid +(19 ,6);
  \draw[line width=2pt] (1,1) -- ++ (0,6);
  \draw[line width=2pt,dashed] (16,1) -- ++ (0,6);

  \foreach \x in {0,...,15} \draw (\x+1.5,0.5)  node {$\x$};
\end{tikzpicture}

\begin{tikzpicture}[scale=0.5]
  \path [fill,red] (0.5,6.5)  
  ++(2,0) circle (5pt)
  ++(3,0) circle (5pt)
  ++(3,0) circle (5pt) ;

  \path [fill,blue] (0.5,6.5)  
  ++(3,0) circle (5pt)
  ++(6,0) circle (5pt) 
  ++(3,0) circle (5pt) ;

  \path (17.5,6.5)  node[anchor=west] {$(\rd{3},\bl{4},\bl{1},\bl{2},0)$};

  \path [fill,red] (0.5,5.5)  
  ++(2,0) circle (5pt)
  ++(3,0) circle (5pt)
  ++(3,0) circle (5pt)
  ++(3,0) circle (5pt) ;

  \path [fill,blue] (0.5,5.5)  
  ++(3,0) circle (5pt)
  ++(6,0) circle (5pt) 
  ++(3,0) circle (5pt) ;
  \path (17.5,5.5)  node[anchor=west] {$(\bl{4},\bl{1},\bl{2},0,\rd{4})$};

  \path [fill,red] (0.5,4.5)  
  ++(2,0) circle (5pt)
  ++(3,0) circle (5pt)
  ++(3,0) circle (5pt)
  ++(3,0) circle (5pt) ;

  \path [fill,blue] (0.5,4.5)  
  ++(3,0) circle (5pt)
  ++(6,0) circle (5pt) 
  ++(3,0) circle (5pt) 
  ++(3,0) circle (5pt) ;
  \path (17.5,4.5)  node[anchor=west] {$(\bl{1},\bl{2},0,\rd{4},\bl{5})$};

  \path [fill,red] (0.5,3.5)  
  ++(2,0) circle (5pt)
  ++(3,0) circle (5pt)
  ++(3,0) circle (5pt)
  ++(3,0) circle (5pt) ;

  \path [fill,blue] (0.5,3.5)  
  ++(3,0) circle (5pt)
  ++(3,0) circle (5pt)
  ++(3,0) circle (5pt) 
  ++(3,0) circle (5pt) 
  ++(3,0) circle (5pt) ;
  \path (17.5,3.5)  node[anchor=west] {$(\bl{2},0,\rd{4},\bl{5},\bl{2})$};

  \path [fill,red] (0.5,2.5)  
  ++(2,0) circle (5pt)
  ++(3,0) circle (5pt)
  ++(3,0) circle (5pt)
  ++(3,0) circle (5pt) ;
  \path [fill,blue] (0.5,2.5)  
  ++(3,0) circle (5pt)
  ++(3,0) circle (5pt)
  ++(6,0) circle (5pt) 
  ++(3,0) circle (5pt) ;
  \path (17.5,2.5)  node[anchor=west] {$(0,\rd{4},\bl{5},\bl{2},\bl{3})$};

  \path [fill,black] (0.5,1.5)  
  ++(1,0) circle (5pt);  
  \path [fill,red] (0.5,1.5)  
  ++(2,0) circle (5pt)
  ++(3,0) circle (5pt)
  ++(3,0) circle (5pt)
  ++(3,0) circle (5pt) ;
  \path [fill,blue] (0.5,1.5)  
  ++(3,0) circle (5pt)
  ++(3,0) circle (5pt)
  ++(6,0) circle (5pt) 
  ++(3,0) circle (5pt) ;
  \path (17.5,1.5)  node[anchor=west] {$(\rd{4},\bl{5},\bl{2},\bl{3},1)$};

  \path [fill,blue] (0.5,1.5) ++(0,0) circle (5pt)
  ++(0,1) circle (5pt)
  ++(0,1) circle (5pt)
  ++(0,1) circle (5pt)
  ++(0,1) circle (5pt)
  ++(0,1) circle (5pt);

  \path [fill,red] (-.5,1.5) ++(0,0) circle (5pt)
  ++(0,1) circle (5pt)
  ++(0,1) circle (5pt)
  ++(0,1) circle (5pt)
  ++(0,1) circle (5pt)
  ++(0,1) circle (5pt);

  \path [fill,black] (-1.5,1.5) ++(0,0) circle (5pt)
  ++(0,1) circle (5pt)
  ++(0,1) circle (5pt)
  ++(0,1) circle (5pt)
  ++(0,1) circle (5pt)
  ++(0,1) circle (5pt);

  \draw  (-2,1) grid +(19 ,6);
  \draw[line width=2pt] (1,1) -- ++ (0,6);
  \draw[line width=2pt,dashed] (16,1) -- ++ (0,6);

  \foreach \x in {0,...,15} \draw (\x+1.5,0.5)  node {$\x$};
\end{tikzpicture}
\caption{The $(5,3)$ Maya cycle for the coloured sequence
  $(\bl{4},\rd{3},\bl{1},\bl{2},0)$ and for the Maya cycle 
  $(\rd{3},\bl{4},\bl{1},\bl{2},0)=
  \bs_0(\bl{4},\rd{3},\bl{1},\bl{2},0)$. }
  \label{fig:53cyclic}
\end{figure}
According to \eqref{eq:mufromnu}, the coloured sequence $(\bnu,C)=(\bl{4},\rd{3},\bl{1},\bl{2},0)$ determines
the flip sequence $\bmu= (14,10,5,8,0)$, and yields the Maya cycle
displayed in the first part of Figure~\ref{fig:53cyclic}. The coloured set that defines
Maya diagram $M_{i+1}$ in the cycle is obtained from the coloured set that defines $M_i$ by applying $\pi$ as described in Proposition~\ref{prop:nuCM} and \eqref{eq:piaction} .  The flip sequence $\bmu$ determines the values of the parameters $(a_0,\dots,a_4)$, which
according to \eqref{eq:aimui} become
$(a_0,a_1,a_2,a_3,a_4)=(8,10,-6,16,-34)$.

In principle, $H_M$ would be pseudo-Wronskians \eqref{eq:pWdef2} for
an arbitrary Maya diagram, but having normalized $(\bnu,C)$ to standard
form, all the rational solutions can be expressed in terms of ordinary
Hermite Wronskians, and no generality is lost. In the case of the
choices made above, the sequence of Wronskians is
\begin{align*}
\H_{M_0}(z)&=\Wr(H_1,H_2,H_4,H_7,H_8,H_{11}),\\
\H_{M_1}(z)&=\Wr(H_1,H_2,H_4,H_7,H_8,H_{11},H_{14}),\\
\H_{M_2}(z)&=\Wr(H_1,H_2,H_4,H_7,H_8,H_{10},H_{11},H_{14}),\\
\H_{M_3}(z)&=\Wr(H_1,H_2,H_4,H_5,H_7,H_8,H_{10},H_{11},H_{14}),\\
\H_{M_4}(z)&=\Wr(H_1,H_2,H_4,H_5,H_7,H_{10},H_{11},H_{14}),
\end{align*}
where $H_n=H_n(z)$ is the $n$-th Hermite polynomial. The rational solution to the dressing chain is given by the tuple $(w_0,w_1,w_2,w_3,w_4|a_0,a_1,a_2,a_3,a_4)$, where $a_i$ and $w_i$ are given by \eqref{eq:wsigMi}--\eqref{eq:aimui} as:
\begin{align*}
w_0(z)&=-z+\ddz\Big[\log \H_{M_1}(z) - \log \H_{M_0}(z)\Big],&& a_0=8,\\
w_1(z)&=-z+ \ddz\Big[\log \H_{M_2}(z)-\log \H_{M_1}(z)\Big],&& a_1=10,\\
w_2(z)&=-z+ \ddz\Big[\log \H_{M_3}(z) - \log \H_{M_2}(z)\Big],&& a_2=-6,\\
w_3(z)&=z+ \ddz\Big[\log \H_{M_4}(z) - \log \H_{M_3}(z)\Big],&& a_3=16,\\
w_4(z)&=-z+ \ddz\Big[\log \H_{M_0}(z) - \log \H_{M_4}(z)\Big],&& a_4=-34.
\end{align*}
Finally, Proposition~\ref{prop:wtof} implies that the corresponding rational solution to the $A_4$-\p\ system \eqref{eq:A2nsystem} is given by the tuple $(f_0,f_1,f_2,f_3,f_4|\a_0,\a_1,\a_2,\a_3,\a_4)$, where
\begin{align*}
f_0(z)&=\tfrac13z+\ddz\Big[\log \H_{M_2}(\cc{2}z) - \log \H_{M_0}(cz)\Big],&& \alpha_0=-\tfrac43,\\
f_1(z)&=\tfrac13z+\ddz\Big[\log \H_{M_3}(cz) -\log \H_{M_1}(cz)\Big],&& \alpha_1=-\tfrac53,\\
f_2(z)&= \ddz\Big[\log \H_{M_4}(cz) - \log \H_{M_2}(cz)\Big],&& \alpha_2=1,\\
f_3(z)&=\ddz\Big[\log \H_{M_0}(cz) -\log \H_{M_3}(cz)\Big],&& \alpha_3=-\tfrac83,\\
f_4(z)&=\tfrac13z+\ddz\Big[\log \H_{M_1}(cz) - \log \H_{M_4}(cz)\Big],&& \alpha_4=\tfrac{17}{3},
\end{align*}
with $c^2=-\tfrac16$.

\end{example}

\section{Connection with the symmetry group approach}
\label{sec:sym}

Noumi and Yamada showed \cite{noumi1998higher} that system
\eqref{eq:A2nsystem} is invariant under a symmetry group, which acts
by B\"acklund transformations on a tuple of functions and
parameters. This symmetry group is the extended affine Weyl group
$\EAWeyln$, generated by the operators
$\bpi,\bs_0,\ldots, \bs_{2n}$ whose action on the tuple
$(f_0,\dots,f_{2n} |\alpha_0,\dots,\alpha_{2n})$ is given by:
\begin{gather}
  \label{eq:BT1}
  \bs_i(f_i) = f_i,\quad s_i(f_j) = f_j \mp \frac{\alpha_i}{f_i}\;
  (j=i\pm 1),\quad s_i(f_j) = f_j \; (j\neq i,\; i\pm 1)\\
  \label{eq:BT2}
  \bs_i(\alpha_i)=-\alpha_i,\quad \bs_i(\alpha_j)=\alpha_j+\alpha_i\;
  (j=i\pm 1),\quad \bs_i(\alpha_j) =\alpha_j\; (j\neq i,i\pm 1)\\
 \label{eq:BT3}\bpi(f_j)=f_{j+1},\\
\label{eq:BT4} \bpi(\alpha_j)=\alpha_{j+1}
\end{gather}
where $i,j=0,\dots,2n \mod(2n+1)$.  A direct calculation and
inspection of \eqref{eq:wtof2} serve to establish that the above
B\"acklund transformations correspond to the following 
  transformation of the dressing chain \eqref{eq:wfchain}:
\begin{gather}
  \label{eq:wBT1}
  \bs_i(w_i) = w_i+\frac{a_i}{w_i+w_{i+1}},\quad \bs_i(w_{i+1}) =
  w_{i+1}-\frac{a_i}{w_i+w_{i+1}},\quad s_i(w_j) = w_j \; (n\neq
  i,i+1)\\
  \label{eq:wBT2}
  \bs_i(a_i)=-a_i,\quad \bs_i(a_j)=a_j+a_i\;
  (j=i\pm 1),\quad \bs_i(a_j) =a_j\; (j\neq i,i\pm 1)\\
  \label{eq:wBT3}\bpi(w_j)=w_{j+1},\\
  \label{eq:wBT4} \bpi(a_j)=a_{j+1}
\end{gather}
The two realizations are equivalent, but we will focus mostly on the
latter realization in terms of dressing chains.

\subsection{Symmetries on Maya cycles}

Above, we showed that an oddly coloured sequence $(\bnu,C)$ specifies
a rational solution to the $A_{2n}$-\p\ system \eqref{eq:A2nsystem},
and that, up to integer translations, every rational solution can
be represented by one such sequence. Since the symmetry group of
transformations \eqref{eq:BT1}-\eqref{eq:BT4} preserves the rational
character of the solutions, it must have a well defined action on Maya
cycles $\bM$ and coloured sequences $(\bnu,C)$. We describe this group
action below.

Let $\bM=(M_0,\dots,M_{2n}, M_{2n+1})$ be a $(2n+1,k)$ Maya cycle
with flip sequence $\bmu=(\mu_0,\ldots, \mu_{2n})$. Define
\begin{align}
  \label{eq:piM}
  \bpi(\bM) &=(M_1,\dots,M_{2n},M_{2n+1}, M_1+k), \\
  \label{eq:siM}
  \bs_i(\bM) &=(M_0,\dots,\hat M_{i+1},\dots,M_{2n+1}),\quad i=0,\ldots,
  2n-1,\\
  \label{eq:s2nM}  
  \bs_{2n}(\bM) &=(\hM_{0},M_1,\dots,M_{2n},\hM_{0}+k),
\end{align}
where
  \[
        \hM_i =   \phi_{\mu_i}(M_{i+1}),\quad i=0,\ldots, 1,\ldots, 2n \]

\noindent
It is clear by inspection that $ \bpi (\bM) $ and
$\bs_i(\bM),\; i=0,\dots,2n$ are Maya cycles with flip sequences given respectively by:

  \begin{align*}
    \pi(\bmu) &= L(\bmu)\\
    \bs_i(\bmu) &= K_{i,i+1}(\bmu),\quad i=0,\ldots, 2n-1\\
    \bs_{2n}(\bmu) &= K_{0,2n}(\bmu),
  \end{align*}
  where $L$ is the circular permutation \eqref{eq:Ldef}, and 
  $K_{i,j}$ denotes a transposition of the indicated elements.  


\begin{prop}\label{prop:WonM}
  The action of $\EAWeyln$ on Maya cycles described in \eqref{eq:piM}
  - \eqref{eq:s2nM} and the action of $\EAWeyln$ by B\"acklund
  transformations \eqref{eq:wBT1} - \eqref{eq:wBT4} is compatible with
  the transformation $\bM \to (\bw|\ba)$  defined in \eqref{eq:wsigMi}
  \eqref{eq:aimui}. 
\end{prop}
\begin{proof}
  The compatibility of $\pi$ follows by a direct inspection.  We
  demonstrate the compatibility of $\bs_0$; the compatibility of the
  other actions is argued similarly. Let $\bM=(M_0,\ldots, M_{2n+1})$
  be a $(2n+1,k)$ Maya cycle, and let $\hbM = \bs_0(\bM)$ as per
  \eqref{eq:siM}.  Let $(\bw,\ba)$ and $(\hat{\bw},\hat{\ba})$ be the
  corresponding rational solutions of the $(2n+1)$-cyclic dressing chain
  as per \eqref{eq:wsigMi} and \eqref{eq:aimui}. Let
  $\bmu=(\mu_0,\ldots, \mu_{2n})$ be the flip sequence corresponding
  to $\bM$.  By construction, $\hbM= (M_0,\hM_1,M_2,\ldots, M_{2n+1})$
  has flip sequence $(\mu_1,\mu_0,\mu_2,\ldots, \mu_{2n})$
  with
  \[ M_1 = \phi_{\mu_0}(M_0),\quad \hM_1 = \phi_{\mu_1}(M_0). \]
  Hence, $\hw_i = w_i, \ha_i = a_i$ for $i=2,\ldots, 2n$.
  Applying \eqref{eq:alphaidef} with
  \[ \lambda_0 = 2\mu_0+1,\quad \lambda_1 = 2\mu_1+1,\quad
    \hat{\lambda}_0 = 2\mu_1+1,\quad \hat{\lambda}_1 =
    2\mu_0+1) \]
  establishes that
  \begin{align*}
    \ha_0 &=2(\mu_0-\mu_1) - a_0 = \bs_0(a_0),\\
    \ha_1 &=    2(\mu_2-\mu_0) = a_0+a_1 = \bs_0(a_1),\\
    \ha_{2n} &=    2(\mu_1-\mu_{2n}) = a_0+a_{2n} = \bs_0(a_{2n})
  \end{align*}

  Observe that
  \[ L_{\hM_1}\stackrel{-\hw_0}\longrightarrow
    L_{M_0}\stackrel{w_0}\longrightarrow
    L_{M_1}\stackrel{w_1}\longrightarrow
    L_{M_2}\stackrel{-\hw_1}\longrightarrow L_{\hM_1}\] forms a cyclic
  factorization chain with shift $0$.
  Hence,
  \begin{equation}
    \label{eq:w0hw0}
    \begin{aligned}
      (-\hw_0+w_0)' + w_0^2-\hw_0^2 &= 2(\mu_1-\mu_0),\\
    (w_0+w_1)' + w_1^2-w_0^2 &= 2(\mu_0-\mu_1),\\
    (w_1-\hw_1)' + \hw_1^2-w_1^2 &= 2(\mu_1-\mu_0),\\
    (-\hw_1-\hw_0)' + \hw_0^2-\hw_1^2 &= 2(\mu_0-\mu_1),\\
    -\hw_0+w_1+w_2-\hw_1 &= 0
    \end{aligned}
  \end{equation}
  It follows by a straightforward elimination that
  \[ (-\hw_0+w_0)(w_0+w_1) = 2(\mu_1-\mu_0) = - a_0\]
  Therefore,
  \[ \bs_0(w_0)=\hw_0 = w_0 + \frac{a_0}{w_0+w_1}.\] A similar
  elimination in \eqref{eq:w0hw0}  serves to show that
  \[ \bs_0(w_1) = \hw_1.\]
\end{proof}

It is also clear that the action of $\pi$ shown in \eqref{eq:piM} is compatible with the B\"acklund transformation \eqref{eq:BT3} \eqref{eq:BT4}.
Note that the action of $\textbf{s}_i $ on the cycle
only changes one Maya diagram $M_{i+1}$, and consequently the
functions $f_{i+1}$ and $f_{i-1}$ are in agreement with
\eqref{eq:BT1}.

We next describe the corresponding action of the symmetry operators
$\pi,\bs_0,\ldots, \bs_{2n}$ on the set of coloured sequences.  Let
$(\bnu,C)\in \Z^{2n+1}_k$ be a coloured sequence.  Define
$\pi(\bnu,c)$ as in   \eqref{eq:piaction}.  Define
\begin{align}
  \label{eq:sinuC}
  \bs_i(\bnu,C) &= (K_{i,i+1}(\bnu), K_{i,i+1}(C)),\quad \bnu\in
  \Z^p,\; C\in (\ZkZ)^p,\\
      \label{eq:s2nnuC}
    \bs_{2n}(\bnu,C)
    &= (K_{2n,0}(\bnu)-\be_0+\be_{2n}, K_{2n,0}(C))\\ \nonumber
    &= (    \nu_{2n}-1,\nu_1,\ldots, \nu_{2n-1}, \nu_0+1, K_{2n,0}(C)) 
\end{align}
where $K_{i,j}$ denotes the transposition of components in positions
  $i$ and $j$.
\begin{prop}
  \label{prop:M+1coords}
  The actions \eqref{eq:piM}-\eqref{eq:siM} of
  $\bpi,\bs_0,\ldots, \bs_{2n}$ on Maya cycles and the corresponding
  actions on coloured seqences \eqref{eq:piaction}
  \eqref{eq:sinuC}\eqref{eq:s2nnuC} satisfy the defining relations of
  the extended affine Weyl group of type $\EAWeyln$:
  \begin{equation}\label{eq:group}
    \bs_i ^2\equiv 1,\quad (\bs_i \bs_{i+1})^{2n+1}\equiv 1,\quad \bpi \bs_i
    \equiv\bs_{i+1}\bpi,\quad \bpi^{2n+1}\equiv 1,
  \end{equation}
  where $\equiv$ indicates equality modulo  translations.
\end{prop}
\begin{proof}
  The above relations follow directly from the relevant definitions.
\end{proof}
\begin{prop}\label{prop:BTaction}
  The action of $\EAWeyln$ on coloured sequences given by
  \eqref{eq:piaction} \eqref{eq:sinuC} \eqref{eq:s2nnuC} and the
  action of $\EAWeyln$ on Maya cycles described in \eqref{eq:piM}
  \eqref{eq:siM} \eqref{eq:s2nM} are compatible with the
  transformation $(\bnu,C)\to \bM$ defined in \eqref{eq:Mipi}.
\end{prop}
\begin{proof}
  Because both sets of actions satisfy \eqref{eq:group} it suffices to
  establish the compatibility of $\bpi$ and $\bs_0$. By inspection of
  \eqref{eq:Mipi},  the
  compatibility of $\bpi$ follows from the relation
  \[ \bpi^{2n+1}(\bnu,C) = T^{2n+1}(\bnu,C) = (\bnu+1,C).\]
  The compatibility of $\bs_0$ follows directly from Proposition
  \ref{prop:nuCM} and the definitions of $\bs_0(\bM)$ and
  $\bs_0(\bnu,C)$.
\end{proof}
\noindent
Note that the above actions preserve the colouring signature. It
follows that $\bpi,\bs_0,\ldots,\bs_{2n}$ preserve the set of oddly
coloured sequences.

According to Proposition~\ref{prop:MM+1} both $\bM$ and $\bM+j$ for any $j\in\Z$ describe the same rational solution of the dressing chain, so the relations \eqref{eq:group} are strict identities for the group action on rational solutions \eqref{eq:BT1}-\eqref{eq:BT4}.
The coincidence of the two representations given by \eqref{eq:BT1} and
\eqref{eq:siM} (with \eqref{eq:wsigMi} and \eqref{eq:wtof2}), entails
interesting identities between Hermite Wronskians, that shall be
further explored elsewhere.

\subsection{Seed solutions}

The usual approach to constructing rational solutions to  the $A_{2n}$-\p\ system \eqref{eq:A2nsystem} is to let the symmetry group act on a number of very simple \textit{seed solutions}, to construct the rest of the rational solutions. 

\begin{definition}
  Fix a $k\in \N$ and $n\in \Nz$ and let $2n+1=p_0+\cdots+p_{k-1}$ be
  a composition of $2n+1$ into $k$ odd parts. A seed sequence is a
  coloured sequence $(\bzero,C_\bp)$ where $\bzero\in \Z^{2n+1}$ is
  the zero vector and where
  \[ C_\bp= (0^{p_0},1^{p_1},\ldots, ).\]
\end{definition}
All the Maya diagrams in the cycle corresponding to a seed coloured sequence $(\bzero,C_\bp)$ have genus zero. The corresponding seed solutions to the  $A_{2n}$-\p\ system \eqref{eq:A2nsystem} for each signature $\bp$ are given by the following Proposition.

\begin{prop}
  Let $2n+1=p_0+\cdots+p_{k-1}$ be a composition of an odd number into
  $k$ odd parts. Let $(f_0,\dots,f_{2n}|\alpha_0,\dots,\alpha_{2n})$
  be the rational solution generated by the corresponding seed
  solution $(\bzero,C_\bp)$.   Then, for every $i=0,\ldots, 2n$ we have
  \begin{align}\label{eq:seed1}
    f_i&= \begin{cases} k^{-1}z & \text{ if } i+1\in Q,\\
      0   & \text{ otherwise}
    \end{cases},\\
    \label{eq:seed2}
    \alpha_i&= \begin{cases} k^{-1} & \text{ if } i+1\in Q,\\
      0 & \text{ otherwise }
    \end{cases},
  \end{align}
  where $Q=\{q_1,\ldots, q_{k}\}$ is the set
  of corresponding partial sums
  \[ q_j=\sum_{r=0}^{j-1} p_r,\quad j=1,\ldots, k.\]
\end{prop}

\begin{proof}
The proof follows by a straightforward application of the construction rules
\eqref{eq:wsigMi}-\eqref{eq:sign} with Proposition~\ref{prop:wtof} to
build the solution $(f_0,\dots,f_{2n}|\alpha_0,\dots,\alpha_{2n})$ corresponding to the  Maya cycle specified by $(\bzero,C_\bp)$.
\end{proof}

\begin{example}\label{ex:seed}

Consider the $(5,3)$ seed solution of the  $A_{4}$-\p\ system corresponding to the
composition $5=1+3+1$, i.e. with signature $\bp=(1,3,1)$. The seed colour sequence is thus
$(0,\rd{0,0,0},\bl{0})$ , and the Maya cycle generated by this sequence is shown in Figure \ref{fig:seed131}.  Note that all the Maya diagrams $M_i$ in the cycle have genus zero. For signature $\bp=(1,3,1)$ we have $Q=\{ 1,4,5 \}$ and the corresponding seed solution given by \eqref{eq:seed1}-\eqref{eq:seed2} is shown in Figure \ref{fig:seed131}.

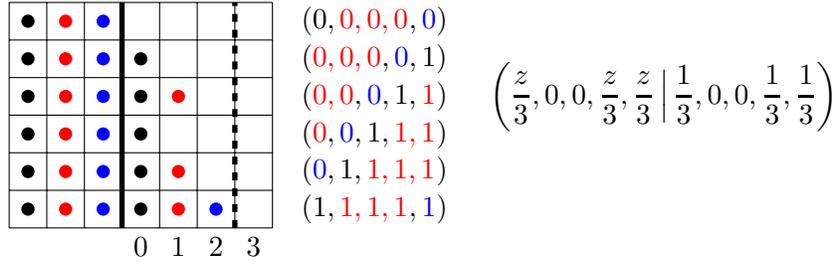
\begin{figure}[h]
  \centering

\begin{tikzpicture}[scale=0.5]
  \path (5.5,6.5)  node[anchor=west] {$(0,\rd{0,0,0},\bl{0})$};

  \path [fill,black] (0.5,5.5)  
  ++(1,0) circle (5pt);

  \path (5.5,5.5)  node[anchor=west] {$(\rd{0,0,0},\bl{0},1)$};

  \path [fill,black] (0.5,4.5)  
  ++(1,0) circle (5pt);

  \path [fill,red] (0.5,4.5)  
  ++(2,0) circle (5pt);

  \path (5.5,4.5)  node[anchor=west]  {$(\rd{0,0},\bl{0},1,\rd{1})$};

  \path [fill,black] (0.5,3.5)  
  ++(1,0) circle (5pt);
  \path (5.5,3.5)  node[anchor=west]   {$(\rd{0},\bl{0},1,\rd{1,1})$};

  \path [fill,black] (0.5,2.5)  
  ++(1,0) circle (5pt);

  \path [fill,red] (0.5,2.5)  
  ++(2,0) circle (5pt);

  \path (5.5,2.5)  node[anchor=west]   {$(\bl{0},1,\rd{1,1,1})$};

  \path [fill,black] (0.5,1.5)  
  ++(1,0) circle (5pt);

  \path [fill,red] (0.5,1.5)  
  ++(2,0) circle (5pt);
  \path [fill,blue] (0.5,1.5)  
  ++(3,0) circle (5pt);
  \path (5.5,1.5)  node[anchor=west]    {$(1,\rd{1,1,1},\bl{1})$};

  \path (10.5,4.5) node[anchor=west]
  {$\displaystyle \left(\frac{z}3,0,0,\frac{z}3,\frac{z}3 \,\Big| \,
      \frac{1}3,0,0,\frac{1}3,\frac{1}3 \right)$};

  \path [fill,blue] (0.5,1.5) ++(0,0) circle (5pt)
  ++(0,1) circle (5pt)
  ++(0,1) circle (5pt)
  ++(0,1) circle (5pt)
  ++(0,1) circle (5pt)
  ++(0,1) circle (5pt);

  \path [fill,red] (-.5,1.5) ++(0,0) circle (5pt)
  ++(0,1) circle (5pt)
  ++(0,1) circle (5pt)
  ++(0,1) circle (5pt)
  ++(0,1) circle (5pt)
  ++(0,1) circle (5pt);

  \path [fill,black] (-1.5,1.5) ++(0,0) circle (5pt)
  ++(0,1) circle (5pt)
  ++(0,1) circle (5pt)
  ++(0,1) circle (5pt)
  ++(0,1) circle (5pt)
  ++(0,1) circle (5pt);

  \draw  (-2,1) grid +(7 ,6);
  \draw[line width=2pt] (1,1) -- ++ (0,6);
  \draw[line width=2pt,dashed] (4,1) -- ++ (0,6);

  \foreach \x in {0,...,3} \draw (\x+1.5,0.5)  node {$\x$};
\end{tikzpicture}

\caption{The Maya cycle for the $(1,3,1)$ seed solution. }
\label{fig:seed131}
\end{figure}

Applying the symmetry operators $\mathbf{s}_0$ and $\mathbf{s}_4$ on the seed colour sequence $(0,\rd{0,0,0},\bl{0})$  leads to
\[ \mathbf{s}_0 (0,\rd{0,0,0},\bl{0}) = (\rd{0},0,\rd{0,0},\bl{0}),\qquad  \mathbf{s}_4 (0,\rd{0,0,0},\bl{0}) =(\bl{-1},\rd{0,0,0},1)  \]
as specified by the action \eqref{eq:sinuC}-\eqref{eq:s2nnuC}. The corresponding cycles and rational solutions of the 
 $A_{4}$-\p\ system are shown in Figure \ref{fig:s0s4seed}. It can be readily verified that the rational solutions are the same that result from the action of the B\"acklund transformations \eqref{eq:BT1}-\eqref{eq:BT4}  on the seed solution.

\end{example}

\begin{figure}[h]
  \centering
\begin{tikzpicture}[scale=0.5]
  \path (5.5,6.5)  node[anchor=west] {$(\rd{0},0,\rd{0,0},\bl{0})$};

  \path [fill,red] (0.5,5.5) ++(2,0) circle (5pt);

  \path (5.5,5.5)  node[anchor=west] {$(0,\rd{0,0},\bl{0},\rd{1})$};

  \path [fill,black] (0.5,4.5)  ++(1,0) circle (5pt);

  \path [fill,red] (0.5,4.5)   ++(2,0) circle (5pt);

  \path (5.5,4.5)  node[anchor=west]  {$(\rd{0,0},\bl{0},\rd{1},1)$};

  \path [fill,black] (0.5,3.5)  
  ++(1,0) circle (5pt);
  \path (5.5,3.5)  node[anchor=west]   {$(\rd{0},\bl{0},\rd{1},1,\rd{1})$};

  \path [fill,black] (0.5,2.5)  
  ++(1,0) circle (5pt);

  \path [fill,red] (0.5,2.5)  
  ++(2,0) circle (5pt);

  \path (5.5,2.5)  node[anchor=west]   {$(\bl{0},\rd{1},1,\rd{1,1})$};

  \path [fill,black] (0.5,1.5)  
  ++(1,0) circle (5pt);

  \path [fill,red] (0.5,1.5)  
  ++(2,0) circle (5pt);
  \path [fill,blue] (0.5,1.5)  
  ++(3,0) circle (5pt);
  \path (5.5,1.5)  node[anchor=west]    {$(\rd{1},1,\rd{1,1},\bl{1})$};

  \path [fill,blue] (0.5,1.5) ++(0,0) circle (5pt)
  ++(0,1) circle (5pt)
  ++(0,1) circle (5pt)
  ++(0,1) circle (5pt)
  ++(0,1) circle (5pt)
  ++(0,1) circle (5pt);

  \path [fill,red] (-.5,1.5) ++(0,0) circle (5pt)
  ++(0,1) circle (5pt)
  ++(0,1) circle (5pt)
  ++(0,1) circle (5pt)
  ++(0,1) circle (5pt)
  ++(0,1) circle (5pt);

  \path [fill,black] (-1.5,1.5) ++(0,0) circle (5pt)
  ++(0,1) circle (5pt)
  ++(0,1) circle (5pt)
  ++(0,1) circle (5pt)
  ++(0,1) circle (5pt)
  ++(0,1) circle (5pt);

  \draw  (-2,1) grid +(7 ,6);
  \draw[line width=2pt] (1,1) -- ++ (0,6);
  \draw[line width=2pt,dashed] (4,1) -- ++ (0,6);

  \foreach \x in {0,...,3} \draw (\x+1.5,0.5) node {$\x$};

  \path (10.5,4.5) node[anchor=west] {
    $\displaystyle \lp \frac{z}3, -\frac{1}{z} , 0 , \frac{z}3 ,
    \frac1z+\frac{z}3\Big|-\frac13,\frac13,0,\frac13,\frac23\rp$};
\end{tikzpicture}

  \begin{tikzpicture}[scale=0.5]
    \path [fill,black] (0.5,6.5) ++(1,0) circle (5pt);
    \path (5.9,6.5)  node[anchor=west] {$(\bl{-1},\rd{0,0,0},1)$};
    
    \path [fill,black] (0.5,5.5) ++(1,0) circle (5pt);
    
    \path (6.5,5.5)  node[anchor=west] {$(\rd{0,0,0},1,\bl{0})$};

  \path [fill,black] (0.5,4.5)  ++(1,0) circle (5pt);

  \path [fill,red] (0.5,4.5)   ++(2,0) circle (5pt);

  \path (6.5,4.5)  node[anchor=west]  {$(\rd{0,0},1,\bl{0},\rd{1})$};

  \path [fill,black] (0.5,3.5)  
  ++(1,0) circle (5pt);
  \path (6.5,3.5)  node[anchor=west]   {$(\rd{0},1,\bl{0},\rd{1,1})$};

  \path [fill,black] (0.5,2.5)    ++(1,0) circle (5pt);

  \path [fill,red] (0.5,2.5)    ++(2,0) circle (5pt);

  \path (6.5,2.5)  node[anchor=west]   {$(1,\bl{0},\rd{1,1,1})$};

  \path [fill,black] (0.5,1.5)   ++(1,0) circle (5pt) ++ (3,0) circle (5pt);

  \path [fill,red] (0.5,1.5)    ++(2,0) circle (5pt);
  \path [fill,blue] (0.5,1.5)  circle (5pt);
  \path (6.5,1.5)  node[anchor=west]    {$(\bl{0},\rd{1,1,1},2)$};

  ++(0,1) circle (5pt)
  ++(0,1) circle (5pt)
  ++(0,1) circle (5pt)
  ++(0,1) circle (5pt);

  \path [fill,red] (-.5,1.5) ++(0,0) circle (5pt)
  ++(0,1) circle (5pt)
  ++(0,1) circle (5pt)
  ++(0,1) circle (5pt)
  ++(0,1) circle (5pt)
  ++(0,1) circle (5pt);

  \path [fill,black] (-1.5,1.5) ++(0,0) circle (5pt)
  ++(0,1) circle (5pt)
  ++(0,1) circle (5pt)
  ++(0,1) circle (5pt)
  ++(0,1) circle (5pt)
  ++(0,1) circle (5pt);

  \path [fill,blue] (-2.5,1.5) ++(0,0) circle (5pt)
  ++(0,1) circle (5pt)
  ++(0,1) circle (5pt)
  ++(0,1) circle (5pt)
  ++(0,1) circle (5pt)
  ++(0,1) circle (5pt);

  \draw  (-3,1) grid +(9 ,6);
  \draw[line width=2pt] (0,1) -- ++ (0,6);
  \draw[line width=2pt, dashed] (5,1) -- ++ (0,6);

  \foreach \x in {-1,...,3} \draw (\x+1.5,0.5) node {$\x$};

  \path (11.5,4.5) node[anchor=west] {
    $\displaystyle \lp \frac{z}3-\frac1z,0,0,  \frac{z}3+\frac1{z} ,
    \frac{z}3\Big|\frac23,0,0,\frac23,-\frac13\rp$};
\end{tikzpicture}

\caption{Applying $\mathbf{s}_0$ and $\mathbf{s}_4$ to the seed solution for signature $\bp=(1,3,1)$.}
\label{fig:s0s4seed}
\end{figure}
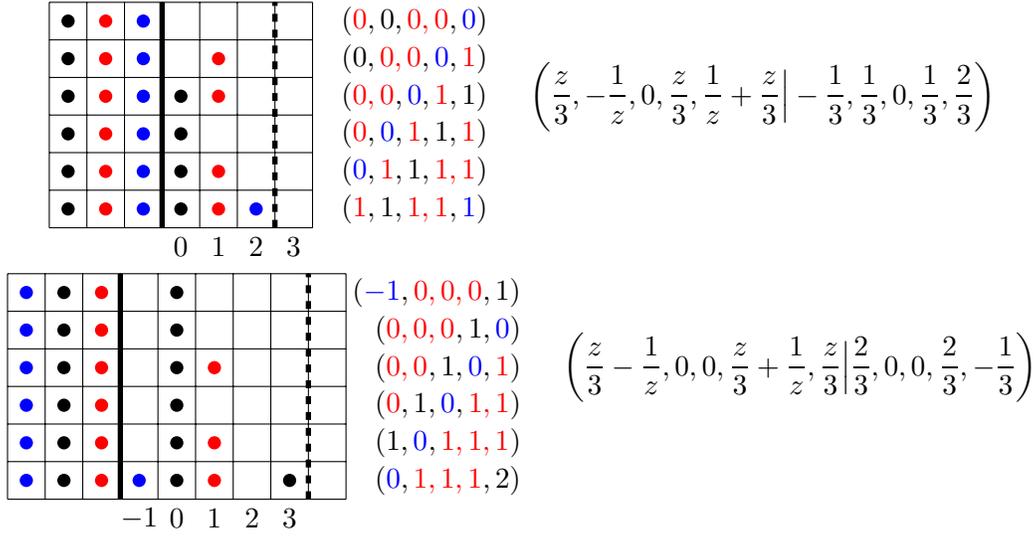

\subsection{Orbits of the extended affine Weyl group $\EAWeyln$ }

In this section we address the transitivity problem, namely to establish that all rational solutions of the $A_{2n}$-\p\ system can be obtained by applying the symmetry group to one of the seed solutions \eqref{eq:seed1}-\eqref{eq:seed2}.

We would like to represent all the rational solutions described by Theorem~\ref{thm:main} as the different orbits of the seed solutions under the action of the symmetry group.

For this purpose, it should be first noted that the action of the
symmetry group \eqref{eq:piM} - \eqref{eq:s2nM} preserves the signature
composition $p=p_0+\cdots+p_{k-1}$. Next, we try to build operators
whose action on block coordinates has a particularly simple action.

\begin{prop}\label{prop:T}
Consider the operators
    \begin{equation}
    \label{eq:Tispi}
    E_i=\bs_i \bs_{i+1} \cdots \bs_{i+2n-1} \pi, \qquad i=0,\dots,2n\mod(2n+1).
  \end{equation}
  The action of these operators on a coloured sequence is given by
  \begin{equation}
    \label{eq:Ti}
    E_i (\bnu,C) = (\bnu+\be_i,C),\quad \bnu\in \Z^{2n+1},\; C\in (\ZkZ)^{2n+1},
  \end{equation}
  where $\be_i\in \Z^{2n+1}$ is the $i\supth$ unit
  vector.

\end{prop}
\begin{proof}
  The proof follows immediately from the action of the operators
  $\bpi,\bs_0,\ldots, \bs_{2n}$ on the coloured sequence $(\bnu,C)$ as
  shown in \eqref{eq:piaction}, \eqref{eq:sinuC},
  \eqref{eq:s2nnuC} .  For $i=0,\ldots, 2n$, we have
  \[ K_{i,i+1} \cdots K_{2n,0} K_{0,1} \cdots K_{i-1,i}L = 1, \]
  where $K_{i,j}$ is the transposition of components $i,j$ and where
  $L$ is the circular permutation \eqref{eq:Ldef}.
  Hence, for $i=1,\ldots, 2n$ we have
  \begin{align*}
    E_i(\bnu,C)
    &= \bs_i \cdots \bs_{2n} \cdots  \bs_{i-1}(L(\bnu)+\be_{2n},L(C))\\
    &= \bs_i \cdots \bs_{2n} (K_{0,1}\cdots
      K_{i-1,i}L(\bnu)+\be_{2n},K_{0,1}\cdots K_{i-1,i}L(C))\\ 
    &= \bs_i \cdots \bs_{2n-1} (K_{2n,0}K_{0,1}\cdots
      K_{i-1,i}L(\bnu)+\be_{2n},K_{2n,0}K_{0,1}\cdots K_{i-1,i}L(C))\\ 
    &= (\bnu+\be_i,C)
  \end{align*}
  For $i=0$ we have
  \begin{align*}
    E_0(\bnu,C)
    &= \bs_0 \cdots \bs_{2n-1} (L(\bnu)+\be_{2n},L(C))\\
    &=  ((K_{0,1}\cdots     K_{2n-1,2n}L)(\bnu)+(K_{0,1}\cdots
      K_{2n-1,2n})\be_{2n},K_{0,1}\cdots  K_{2n-1,2n}L(C))\\ 
    &= (\bnu+\be_0,C)
\end{align*}
\end{proof}

\begin{example}
  Consider the action of $E_1=\bs_1\bs_2\pi$ on the coloured sequence
  $(2,\rd{3},\bl{0})$. Using the transformation rules
  \eqref{eq:piaction}, \eqref{eq:sinuC},  \eqref{eq:s2nnuC}  we
  have
  \[ (2,\rd{3},\bl{0}) \xrightarrow{\bpi} (\rd{3},\bl{0},3)
    \xrightarrow{\bs_2} (2,\bl{0},\rd{4}) \xrightarrow{\bs_1}
    (2,\rd{4},\bl{0}) = (2,\rd{3},\bl{0}) + \be_1.
  \]
\end{example}

Next, we describe the symmetry operators that act on coloured
sequences by direct permutations.
\begin{prop}\label{prop:sigma}
  The operators $\bs_0,\ldots, \bs_{2n-1}$ generate the action of the
  permutation group $\cS_{2n+1}$ on the set of coloured sequences $(\bnu,C)$ of length $2n+1$.
\end{prop}
\begin{proof}
  This follows directly from the definition \eqref{eq:sinuC}.
\end{proof}

%

\begin{thm}\label{thm:symmetry}
Every rational solution of a $(2n+1)$-cyclic dressing chain  can be obtained by the action of  the symmetry group $\EAWeyln$  on a seed solution.

\end{thm}

\begin{proof}
Assume $(\bw|\ba)$ is rational solution of a $(2n+1)$-cyclic dressing chain with shift $\Delta=2k$. By Theorem~\ref{thm:main}, this solution can be indexed by  a coloured sequence $(\bnu,C)\in\Z^{2n+1}_k$ in standard form. Let $\bp=(p_0,\dots,p_{k-1})$ be its signature and consider the seed solution $(\mathbf{0},C_{\bp})$ corresponding to that signature. Since  $(\bnu,C)$ is in standard form, all the components are $\nu_i\geq 0$ for $ i=0,\dots,2n$. There clearly exists a sequence of operators $E_i$ defined in \eqref{eq:Tispi} that map the seed sequence $(\mathbf{0},C_{\bp})$ into a seed sequence $(\bnu',C')$ that differs from $(\bnu,C)$ at most  by a permutation of its elements, i.e. such that $[(\bnu',C')]=[(\bnu,C)]$. By Proposition\ref{prop:sigma}, there is a sequence of symmetry operators that map $(\bnu',C')$ into $(\bnu,C)$.
\end{proof}

Note, however, that given a rational solution $(\bw|\ba)$, the seed solution and sequence of symmetry transformations is not unique. If $(\bnu,C)$ is the coloured sequence in standard form that corresponds to $(\bw|\ba)$, so does the coloured sequence $T(\bnu,C)$ by Propositions~\ref{prop:M+1} and \ref{prop:MM+1}.
But $T(\bnu,C)$ has signature $L^{-1}(\bp)=(p_{k-1},p_0,\dots,p_{k-2})$. We see thus that seed solutions \eqref{eq:seed1}-\eqref{eq:seed2} corresponding to circular permutations of a given signature belong to the same orbit under the symmetry group $\EAWeyln$.

\begin{remark}
Similar formulas for the action of the generators of the extended affine Weyl group $\EAWeyln$ on Maya diagrams have been given by Noumi in his book (see \cite{noumi2004painleve} Ch. 7.6). However, the main difference between Proposition 7.12 in  \cite{noumi2004painleve}  and Proposition~\ref{prop:WonM} in this work is that in the former the symmetry group acts on a single Maya diagram, while the latter treats the group action on a Maya cycle. Since a rational solution corresponds to a Maya cycle rather than to single Maya diagram (see Theorem~\ref{thm:characterization}), we believe that Propositions~\ref{prop:WonM} and \ref{prop:BTaction} provide the correct corespondence with the action of $\EAWeyln$ on rational solutions \eqref{eq:BT1}-\eqref{eq:BT4} via B\"acklund transformations.
\end{remark}

\section{Degenerate solutions}\label{sec:degen}

Degenerate solutions describe certain embeddings of lower order systems $A_{2m}$-\p\ into higher order systems $A_{2n}$-\p\ for $n>m$. Also, they single out special properties with respect to the isotropy of the group action.

\begin{definition}
Let $(\bw|\ba)=(w_0,\dots,w_{2n}|a_0,\dots,a_{2n})$ be a rational solution to a $(2n+1)$-dressing chain with shift $\Delta=2k$ and let $(\bnu,C)\in\Z^{2n+1}_k$ be its  standard coloured sequence. The flip sequence $\bmu=(\mu_0,\dots,\mu_{2n})\in \Z^{2n+1}$ is given by
\[  \mu_i=k\nu_i+C_i,\qquad i=0,\dots,2n.\ \]
The rational solution $(\bw|\ba)$ is said to be a \textit{degenerate solution} if at least one of the following conditions hold
\begin{enumerate}
\item The flip sequence $\bmu$ contains repeated elements.
\item $\mu_{2n}=\mu_{0}+k$.
\end{enumerate}
\end{definition}

From a degenerate solution one can immediately build a solution to a
lower order system in the following manner. Consider for
simplicity that condition (1) above holds and $\bmu$ has repeated
elements. This means that $[(\bnu,C)]$ is a multiset but not a set.
Let $(\nu_i,C_i)$ and $(\nu_j,C_j)$ be two elements of the coloured
sequence $(\bnu,C)$ whose value and colour coincide. Then the coloured
sequence $(\tilde\bnu,\tilde C)$ obtained by dropping these two
elements is an oddly $k$-coloured sequence of length $2n-1$ and thus
defines a solution to the $A_{2n-2}$-\p\ system.

If the two equal elements  in the coloured sequence are consecutive, e.g. $(\nu_i,C_i)=(\nu_{i+1},C_{i+1})$, then two consecutive flips $\mu_i=\mu_{i+1}$ happen at the same site and the Maya cycle $\bM=(M_0,\dots,M_{2n+1})$ contains two identical Maya diagrams $M_i=M_{i+2}$. Correspondingly, the Maya cycle $\tilde{\bM}$ obtained by dropping $M_{i+1}$ and $M_{i+2}$ is $(2n-1,k)$-cyclic:
\[ \bM=(M_0,\dots,M_i,M_{i+1},M_{i+2},\dots ,M_{2n+1})\to \tilde{\bM}=(M_0,\dots,M_i,M_{i+3},\dots,M_{2n+1}). \]
 If $(f_0,\dots,f_{2n}|\alpha_0,\dots,\alpha_{2n})$ is the corresponding degenerate solution to the  $A_{2n}$-\p\ system, then we can build a solution $(\tilde f_0,\dots,\tilde f_{2n-2}|\tilde \alpha_0,\dots,\tilde \alpha_{2n-2})$    to the $A_{2n-2}$-\p\ system by setting
 \begin{equation}\label{eq:embed}
 \begin{aligned}
\tilde f_j=\begin{cases}
f_j &\text{ if } j\leq i-2,\\
f_j+f_{j+2} &\text{ if } j= i-1,\\
f_{j+2} &\text{ if } j\geq i,
 \end{cases}\qquad \tilde \alpha_j=\begin{cases}
 \alpha_j &\text{ if } j\leq i-2,\\
 \alpha_j+ \alpha_{j+2} &\text{ if } j= i-1,\\
 \alpha_{j+2} &\text{ if } j\geq i,
 \end{cases}
 \end{aligned}
 \end{equation}
 where the indices for $f_i,\alpha_i$ are taken $\mod{2n+1}$ and those for $\tilde f_i,\tilde \alpha_i$ are taken $\mod{2n-1}$. It is clear also that $f_i(z)=0$ and $\alpha_i=0$ in the degenerate solution.
 
 We see thus that in the case of the consecutive flips at the same site leads to a rather trivial embedding of the lower order system into the higher order one. However, there could be two repeated elements in the flip sequence which are not consecutive, i.e $(\nu_i,C_i)=(\nu_{j},C_{j})$ for $j>i+1$. It is still true that the coloured sequence $(\tilde\bnu,\tilde C)$ obtained by dropping these two repeated elements will define a solution to the $A_{2n-2}$-\p\ system, but it is no longer true that one can simply eliminate two Maya diagrams from the $(2n+1,k)$ Maya cycle $\bM=(M_0,\dots,M_{2n+1})$ to obtain a $(2n-1,k)$ Maya cycle $\tilde{\bM}$. This case leads to nontrivial embeddings, in which the reduced solution $(\tilde f_0,\dots,\tilde f_{2n-2}|\tilde \alpha_0,\dots,\tilde \alpha_{2n-2})$   cannot be obtained by a linear combination of the degenerate solution of the higher order system $(f_0,\dots,f_{2n}|\alpha_0,\dots,\alpha_{2n})$.

Finally, the following Proposition states that degenerate solutions have a nontrivial isotropy group.
\begin{prop}\label{prop:isometry}
  If $(f_0,\dots,f_{2n}|\alpha_0,\dots,\alpha_{2n})$ is a degenerate rational solution of the $A_{2n}$ system, then there is an element of the symmetry group $\EAWeyln$ that leaves it invariant.
\end{prop}
\begin{proof}
This follows directly from Theorem~\ref{thm:main} and the action \eqref{eq:sinuC}-\eqref{eq:s2nnuC} of the symmetry group $\EAWeyln$ on coloured sequences. In the case where $(\nu_i,C_i)=(\nu_{i+1},C_{i+1})$, the rational solution is a fixed point of the generator $\mathbf{s}_i$. Otherwise, there is a sequence of symmetry transformations that performs the transposition between the two identical elements of the coloured sequence.
\end{proof}

\begin{example}
  We illustrate here the occurence of degenerate solutions and the
  difference between trivial and non-trivial embeddings. Consider the
  degenerate solutions of the $A_4$-\p\ system corresponding to the
  following coloured sequences:
\begin{eqnarray}
  (\bnu^{(1)},C^{(1)})&=& (0,\rd{1},1,1,\bl{0} )\qquad \bmu^{(1)}=(0,4,3,3,2),\\
  (\bnu^{(2)},C^{(2)})&=& (0,1,\rd{1},1,\bl{0} ),\qquad \bmu^{(2)}=(0,3,4,3,2).
\end{eqnarray}
 Both solutions are degenerate because the flip sequence contains repeated elements, but in the first case they are consecutive while in the second the are not. The corresponding rational solutions are
{\small
\begin{eqnarray*}
  (\bdf^{(1)} |\balpha^{(1)}) &=& \left( \frac{z}3-\frac{1}z +\frac{2z}{z^2+3}, \frac{z}3 + \frac{2 z}{ z^2-3},0,\frac{1}z-\frac{2z}{z^2-3},\frac{z}3-\frac{2z}{z^2+3}\, \Bigg| \,\frac{4}3,-\frac{1}3,0,- \frac{1}3,\frac{1}3   \right) ,\\
  (\bdf^{(2)} |\balpha^{(2)})&=& \left( \frac{z}3, \frac{z}3 + \frac{2 z}{ z^2-3},-\frac{1}z+\frac{2z}{z^2+3} ,-\frac{1}z-\frac{2z}{z^2-3},\frac{z}3-\frac{2z}{z^2+3}\, \Bigg| \,1,\frac{1}3,- \frac{1}3,- \frac{1}3,\frac{1}3   \right)
\end{eqnarray*}
} If we drop the two repeated elements of the coloured sequence, we
obtain in both cases the following solution of the $A_2$-\p\ system.

\[
\begin{aligned}
&(\tilde{\bnu},\tilde C)= (0,\rd{1},\bl{0} ),\qquad \tilde{\bmu}=(0,4,2),\\
&(\tilde{\bdf} |\tilde{\balpha}) = \left(\frac{z}3-\frac{1}z+\frac{2z}{z^2+3},\frac{z}3+\frac{1}z,\frac{z}3-\frac{2z}{z^2+3}  \, \Bigg|\, \frac{4}3,-\frac{2}3,\frac{1}3 \right)
\end{aligned}
\]

This solution to the $A_2$-\p\ system is clearly obtained from
$(\bdf^{(1)} |\balpha^{(1)})$ by applying \eqref{eq:embed}. It
is also clear from \eqref{eq:BT1}-\eqref{eq:BT2} that
\[\mathbf{s}_2 (\bdf^{(1)} |\balpha^{(1)})=(\bdf^{(1)}
  |\balpha^{(1)}),\]
  as explained in Proposition~\ref{prop:isometry}.

\end{example}

\section{Summary and Outlook}\label{sec:summary}

This paper provides a complete classification of the rational
solutions of Painlev\'e $\Pfour$ and its higher order hierarchy known
as the $A_{2n}$-\p\ or Noumi-Yamada system. First, we recall the
equivalence between the Noumi-Yamada system \eqref{eq:A2nsystem} and a
cyclic dressing chain of Schr\"odinger operators. Then, we show by a
careful investigation of the local expansions of the rational
solutions around their poles, that the solutions have trivial
monodromy, and therefore they must be expressible in terms of
Wronskian determinants whose entries are Hermite polynomials. Next, we
use Maya diagrams to classify all the $(2n+1)$-cyclic dressing chains
and therefore achieve a complete classification. Finally, we connect
our results with the geometric approach mastered by the japanese
school, showing a representation for the action of the symmetry group
of B\"acklund transformations in terms of Maya cycles and oddly
coloured integer sequences.

The natural extension of this work is to tackle the full
classification of the rational solutions to the $A_{2n+1}$-\p\
systems, which include $\Pfive$ and its higher order
extensions. Although the analysis is considerably more involved, it
should be possible to extend this approach from the odd-cyclic to the
even-cyclic case. In this regard, we would like to formulate the
following
\begin{conj}
  All rational solutions of $\Pfive$ (and the $A_{2n+1}$-\p\ system)
  can be expressed as Wronskian determinants whose entries involve
  Laguerre polynomials.
\end{conj}
We have solid evidence to believe that this conjecture is true. Given
the correspondence of $\Pfour$ with Wronskians of Hermite polynomials
and $\Pfive$ with Wronskians of Laguerre polynomials, it is very
appealing to suggest also that Wronskians of Jacobi polynomials must
play an important role in the description of the rational solutions to
$\Psix$. However, unlike $\Pfour$ and $\Pfive$, in the case of $\Psix$
we still lack a basic theory of dressing chains and Darboux
transformations.  Some ideas in this direction can be found in
\cite{adler1994modification}. Also, it would be interesting to investigate
combinatorial representations for the Painlev\'e equations  \PI,\PII,\PIII $ $ and $\Psix$, perhaps considering the realization of these equations as
Fuchs-Garnier systems \cite{jkt09}.

All algebraic solutions of  Painlev\'e's  $\Psix$ equation were
  described in \cite{listyk14} in terms of finite orbits of the
  corresponding symmetry group.  It would be interesting to consider
  whether the approach in this paper can be extended to give an analogous classification of the
  algebraic solutions of \PIV.

\addtocontents{toc}{\protect\setcounter{tocdepth}{0}}
\section*{Acknowledgements}

The research of DGU has been supported in part by the Spanish MICINN under grants PGC2018-096504-B-C33 and RTI2018-100754-B-I00,  the European Union under the 2014-2020 ERDF Operational Programme and the Department of Economy, Knowledge, Business and University of the Regional Government of Andalusia (project FEDER-UCA18-108393). We would like to thank three anonymous referees for their valuable comments that helped to improve the final version of this manuscript.


\end{document}